\documentclass{IEEEtran}
\usepackage{amsmath,amsfonts}
\usepackage{algorithmic}
\usepackage{algorithm}
\usepackage{lineno,hyperref}
\usepackage{multirow}
\usepackage{epstopdf}
\usepackage{graphicx}
\usepackage{float}
\usepackage{caption}
\usepackage{amssymb}
\usepackage{hyperref}
\usepackage{changes}
\usepackage{color,array,amsthm}
\usepackage{graphicx}

\newtheorem{theorem}{Theorem}
\newtheorem{lemma}{Lemma}

\newtheorem{remark}{Remark}
\newtheorem{assumption}{Assumption}
\usepackage[inkscapelatex=false]{svg}
\usepackage[numbers,sort&compress]{natbib}
\usepackage{soul} % 导入 soul 包
\usepackage{color, xcolor} % 颜色包，color 必须导入，xcolor 建议导入

\soulregister{\cite}7 % 注册\cite命令
\soulregister{\citep}7 % 注册\citep命令
\soulregister{\citet}7 % 注册\citet命令
\soulregister{\ref}7 % 注册\ref命令
\soulregister{\pageref}7 % 注册\pageref命令
\modulolinenumbers[5]

\begin{document}

\title{Achieving Precisely-Assigned Performance Requirements for Spacecraft Attitude Control
}

\author{Jiakun Lei
%	Tao Meng,
%	Kun Wang,
%	Weijia Wang
	\thanks{Jiakun Lei, Ph.D. Candidate, School of Aeronautics and Astronautics, Zhejiang University, Hangzhou, China, 310027, Email:leijiakun@zju.edu.cn}
	
%	\thanks{Tao Meng, Prof., School of Aeronautics and Astronautics, Zhejiang University, Hangzhou, China, 310027;
%		Hainan Institute of Zhejiang University, Sanya, China, 572025, 
%		Email: mengtao@zju.edu.cn}
	
	%	\thanks{Yang Zhu, Prof., College of Control Science and Engineering, Zhejiang University, Hangzhou, China, 310027,  (Email: zhuyang88@zju.edu.cn).}
	
%	\thanks{Kun Wang, Ph.D. Candidate, School of Aeronautics and Astronautics, Zhejiang University, Hangzhou, China, 310027, Email:Wang\_kun@zju.edu.cn}
%	
%		\thanks{Weijia Wang, Ph.D. Candidate, School of Aeronautics and Astronautics, Zhejiang University, Hangzhou, China, 310027, Email:weijiawang@zju.edu.cn}
}

\maketitle

\begin{abstract} 
This paper investigates the attitude control problem of spacecraft, with the objective of achieving precise performance criteria including precise settling time, steady-state error, and overshoot elimination. To tackle this challenge, we propose the Precisely-Assigned Performance (PAP) control scheme. Firstly, we utilize a parameterized function to explicitly characterize a reference for the transient responses, termed the Reference Performance Function (RPF). Subsequently, leveraging the concept of the RPF, we define a performance-satisfied tube region and introduce the concept of control barrier functions to derive a sufficient condition for the state trajectory to converge and remain confined within this tube region. By  introducing the concept of Sontag's universal formula for stabilization, a PAP controller, constructed based on the backstepping method, is then designed to guide the system to satisfy these affine constraint conditions, and a disturbance observer is further integrated to handle perturbations. Theoretical proofs are presented to demonstrate the controller's capability to establish the boundedness of the overall system and ensure that each state trajectory will converge into the performance-satisfied region within a finite time duration under any conditions. Finally, numerical simulation results are presented to validate the effectiveness of the proposed method.
\end{abstract}


\section{Introduction}
Contemporary space missions, including on-orbit servicing, deep-space exploration, and on-orbit imaging tasks, often require high-performance attitude control. For example, during an on-orbit imaging task, the spacecraft's attitude control system should not only provide a rapid attitude maneuver to align the payload's boresight axis with a given direction but also ensure that the intersection point of the boresight axis with the ground precisely follows specified curves. This goal essentially requires the pointing error to converge just around a specified time instant, while the steady-state error should remain confined within a designated error bound thereafter. These requirements impose intricate restrictions on the closed-loop system's transient response. Motivated by this topic, this paper investigates the problem of spacecraft attitude control with the aim of achieving multiple precisely-assigned performance criteria, including settling time, steady-state error bound, and overshoot elimination.

Regarding the goal of ensuring fast convergence in the attitude error system, the finite-time \cite{hu2017continuous,zou2019finite}, fixed-time \cite{zou2020fixed,huang2018fixed}, and predefined-time stability \cite{sanchez-torres_discontinuous_2014,xie_adaptive_2022,wang2020attitude} techniques are commonly employed. These methods construct special structures for the control law to establish specific convergence dynamics based on Lyapunov certificates, thereby guiding the error states to converge while satisfying time-related requirements.
Despite their efficiency, these methods can accelerate the convergence rate of the error system, but they do not guarantee that the system's steady-state error will be upper-bounded by a given constant.

For the control methodology aimed at achieving specific performance criteria, which involves requirements on both the transient response, such as convergence time, and the steady-state response, such as steady-state error, the Prescribed Performance Control (PPC) scheme has gained significant attention in recent years \cite{bechlioulis_robust_2008,bechlioulis2011robust}. The PPC scheme use performance functions to characterize the desired response. Through the process of homeomorphic error transformation and the design of a control law to stabilize the transformed error states, each state trajectory can be ensured to be confined within a funnel-like performance envelope, enclosed by a pair of performance functions, thereby achieving the desired performance criteria throughout the entire process \cite{wei_2021_overview,bu2023prescribed}.

Owing to its efficiency, the PPC control scheme has been applied in numerous studies to achieve performance-guaranteed control, particularly for spacecraft attitude control \cite{hu_model-free_2021,luo_low-complexity_2018,bu_2018_new,huang_fault-tolerant_2020}. In \cite{wei_learning-based_2019}, the authors devised a dual-layer PPC strategy for the spacecraft formation flying problem. In another study \cite{hu2017adaptive}, the authors introduce an asymmetric performance envelope design and a linear auxiliary variable to theoretically eliminate overshooting. Concurrently, fixed-time, and predefined-time stability techniques have been widely utilized to construct the performance envelope, as referenced in \cite{gao_novel_2021,yin_appointed_2019}.

Although the PPC control scheme effectively achieves performance-guaranteed control, it faces a notable singularity problem, hindering its practical application. As mentioned by \cite{lei2023singularity}, also emphasized in \cite{bu2023prescribed,yong2020flexible}, typical PPC schemes require the invariance of the performance envelope constraint, where error states must be restricted within their corresponding performance envelopes from start to end. However, this constraint is at risk of violation if the system is heavily perturbed. Additionally, limited controller capability can also lead to the occurrence of such a singularity problem, as mentioned and investigated in \cite{yong2020flexible}.
On the other hand, for the problem discussed in this paper, the system must converge with precisely-assigned transient behavior. Achieving such a control objective remains difficult for typical PPC control schemes, as technical barriers arise from mainly two perspectives:

\begin{enumerate}
	\item \textit{Vague Performance Description:} Typical PPC control schemes utilize an "envelope" to bound the error state, thereby preserving unnecessary degrees of freedom in how the error state converges. This approach makes it challenging to precisely characterize the ideal transient response.
	
	\item \textit{Mismatching of Envelope Design and Performance:} While assigning a tight performance envelope seems to provide a precise characterization of the desired transient response, this is challenging for a typical PPC control scheme to do so. Since PPC control schemes rely on the gradient of their envelope bounds to guide the system's error states, this leads to an overly strong repulsion effect when the error state approaches its corresponding performance bound. As a result, the error state may be influenced by the lower and upper bounds repeatedly if the performance envelope becomes too narrow, causing the chattering problem \cite{lei2023singularity,wang2022finite}. Due to this drawback, it is difficult to assign a performance envelope directly according to the given performance requirements, and hence there still lacks an explicit relationship between the envelope design and its actual reachable performance. Such a problem has also been mentioned in \cite{bu2023prescribed}.
	
\end{enumerate}

Recently, the development of Control Barrier Function (CBF) has gained much attention in the field of safety-critical control problems, as stated by much literature \cite{ames2019control,ames2016control,xiao2019control}. The concept of CBF is originated from Nagumo's Forward Invariant Theorem \cite{nagumo1942lage,blanchini1999set,horvath2016invariance}, where the authors provide the necessary and sufficient condition for the system to remain confined in a set by considering the time-derivative of a smooth scalar function. The CBF methods have been combined with optimization techniques to calculate a control input to directly render the safe set invariant, known as CBF-QP \cite{ames2014control}, showing its efficiency in various safety-related control problems, such as robotics control \cite{rauscher2016constrained,ferraguti2022safety}, vehicle control \cite{alan2023control}, and spacecraft control \cite{breeden2021guaranteed,breeden2023robust}.

\textit{Motivation and Contribution Statement: }
For the discussed control problem, due to these mentioned problems, the system is hard to exhibit the desired response as characterized by the performance envelope with the PPC control schemes. Meanwhile, the intrinsic singularity problem often leads to unexpected chattering when coupled with external perturbation issues, threatening the system's stability. Motivated by the concept of CBF, we migrate the concept of "control to ensure safety" into the realm of "control to satisfy performance requirements". In fact, such an idea was previously investigated in \cite{verginis2022funnel}, where the authors utilize the CBF-QP technique to numerically solve a control input for an affine control system. However, the typical design of the funnel does not meet our control objectives due to the vague description of performance criteria. Also, the utilization of Quadratic Programming (QP) increases the computational burden of the system. Consequently, it still remains necessary and challenging to develop a unified and simple control framework to address our control objective. Accordingly, the main contributions of this paper are stated as follows:

\begin{enumerate}
	\item A complete control scheme is presented for the problem under discussion, effectively achieving the precisely assigned performance criteria, referred to as Precisely-Assigned Performance Control (PAP). The PAP control scheme introduces the concept of the Reference Performance Function (RPF), which uses parameterized functions to precisely characterize the desired transient response. By guiding each state trajectory to converge and remain within a "tube" centered around the RPF, the desired performance criteria can be achieved.
	
	\item The proposed PAP control scheme departs from the structure of typical PPC controllers, which are generally constructed based on the barrier condition theorem \cite{ames2019control}. Instead, it employs the concept of Nagumo's forward invariance, extending the concept of "control to satisfy safety" to "control to satisfy performance requirements". The introduction of CBF naturally renders the proposed method singularity-free.
	
	\item Building on the concept of control barrier functions, we establish sufficient conditions for each state trajectory to converge and remain confined within its corresponding performance-satisfied tube region. Since these conditions are affine constraints, this motivates us to introduce Sontag's universal formula for stabilization into this problem. Such an approach not only provides a simple controller structure that can be regarded as backstepping control with a time-varying gain, thereby reducing the computation burden. Concurrently, it ensures the globally rapid attraction of each state trajectory to its RPF curve within an explicitly expressed time duration.
\end{enumerate}

\textit{Notations: } 
This paper introduces the following notations for analysis. The norm of any vector or matrix is denoted by $\|\cdot\|$. For an arbitrary vector $\boldsymbol{a}\in\mathbb{R}^{3}$, $\boldsymbol{a}^{\times}$ represents the skew-symmetric matrix of the cross product. The notation $\text{diag}(a_{i})$ signifies the diagonal matrix whose elements are spanned by $a_{i}$. In this context, $\mathfrak{R}_{b}$ refers to the spacecraft body-fixed frame, $\mathfrak{R}_{i}$ refers to the Earth-central inertia frame, and $\mathfrak{R}_{d}$ refers to the target body-fixed frame. Additionally, $\boldsymbol{I}_{N}$ represents the identity matrix in $\mathbb{R}^{N\times N}$; $\boldsymbol{1}_{N}$ represents the $\mathbb{R}^{N}$ column vector, where its elements are all $1$; $\boldsymbol{0}$ represents the zero vector with appropriate dimension.

\section{Preliminaries}
%\begin{definition} \label{defSmooth}
%	\newcommand{\Omegadef}[4]{\Omega(#1, #2, #3, #4)}
%	Define a smooth mollified switching function $\Omega(x, p, S_{0}, S_{1})$, where $x$ stands for the argument and $p$, $S_{0}$ and $S_{1}$ are design parameters. The expression of $\Omega(\cdot)$ is given as follows:
%	\begin{equation}
%		\label{Omegadef}
%		\Omega(x) = \begin{cases}
%			0 &  x \in (-\infty, S_{0}]\\
%			\frac{1}{2}\left[\tanh\frac{p\left(S_{1}-S_{0}\right)\cdot\left(x-S_{m}\right)}{\sqrt{(x-S_{0})(S_{1}-x)}} + 1\right] & x \in (S_{0}, S_{1})\\
%			1 &  x \in [S_{1}, +\infty)\\
%		\end{cases}
%	\end{equation} 
%where $S_{m}$ is a constant defined as $S_{m} \triangleq \frac{1}{2}(S_{0}+S_{1})$; $p>\frac{1}{S_{1}-S_{0}}$ determines the increasing rate of $\Omegadef{x}{p}{S_{0}}{S_{1}}$ with respect to $x$; $S_{0}$, $S_{1}$ stand for the switching segment point of $\Omegadef{x}{p}{S_{0}}{S_{1}}$. It can be observed that the function $\Omegadef{x}{p}{S_{0}}{S_{1}}$ is smooth, which compresses arbitrary input $x$ into a range of $[0,1]$, and hence provides a smooth approximation of $0-1$ switching process.
%\end{definition}

\begin{lemma} \label{lemmatanh}
	Given a sufficiently large constant $h > 0$, $|\tanh(hx)|\ge |x|$ holds for $x\in\left[-h_{0},h_{0}\right]$, where $h_{0} \in\left(0,1\right)$ is a positive constant \cite{li2016adaptive}.
\end{lemma}
\begin{proof}
Without loss of generality, we consider \(x \geq 0\). According to the definition of the hyperbolic tangent function, we have: \(\tanh(hx) = 1 - 2e^{-2hx} + \mathcal{O}(e^{-2hx})\), where \(\mathcal{O}(e^{-2hx})\) stands for the higher-order infinitesimal term.
Thus, \(1 - 2e^{-2hx} \leq \tanh(hx)\) holds for \(x \in [0,+\infty)\), implying that the non-zero root \(x_{c}\) of the equation \(1 - 2e^{-2hx_{c}} = x_{c}\) will be smaller than that of the equation \(\tanh(hx_{p}) = x_{p}\), i.e., \(x_{c} < x_{p}\). Therefore, given arbitrary \(x_{0} \in (0,1)\), one can solve for \(h_{0}\), where \(|\tanh(hx)| > x\) holds for \(x \in [-x_{0},x_{0}]\) if $h \ge h_{0}$.

Here we provide an example to facilitate the explanation. Let \(x_{0} = 0.95\), then we have \(h_{0} = \frac{10\ln(40)}{19} \approx 1.9415\). This implies that we should ensure \(h > 1.9415\), and then \(|\tanh(hq_{evi}(t))| > |q_{evi}(t)|\) will be satisfied for \(q_{evi}(t) \in [-0.95,0.95]\).

\end{proof}

\section{Problem Formulation}
\subsection{System Modeling}\label{errsys}
Considering the general unit quaternion-based attitude error model of rigid-body spacecrafts, let $\boldsymbol{q}_{e} = \left[\boldsymbol{q}^{\text{T}}_{ev},q_{e0}\right]^{\text{T}}\in\mathbb{R}^{4}$ be the attitude quaternion error, where $\boldsymbol{q}_{ev} = \left[q_{ev1},q_{ev2},q_{ev3}\right]^{\text{T}}\in\mathbb{R}^{3}$ and $q_{e0}\in\mathbb{R}$ represents the vector and the scalar part of $\boldsymbol{q}_{e}$, respectively.
 Let $\boldsymbol{\omega}_{e}\in\mathbb{R}^{3}$ be the attitude angular velocity error, resolved in the body-fixed frame $\mathfrak{R}_{b}$. Applying these notations, the attitude error model can be expressed as follows \cite{wertz_2012_spacecraft}:
\begin{equation}\label{errsystem}
	\begin{aligned}
			\dot{\boldsymbol{q}}_{ev} &= \boldsymbol{F}_{e}\boldsymbol{\omega}_{e},\quad \dot{q}_{e0}  = -\frac{1}{2}\boldsymbol{q}^{\text{T}}_{ev}\boldsymbol{\omega}_{e}\\
				\boldsymbol{J}\dot{\boldsymbol{\omega}}_{e} &= \boldsymbol{\Omega}_{e} + \boldsymbol{u} + \boldsymbol{d}
	\end{aligned}
\end{equation}
In the above expression, $\boldsymbol{\omega}_{e}$ can be decomposed as $\boldsymbol{\omega}_{e} = \boldsymbol{\omega}_{s} - \boldsymbol{C}_{e}\boldsymbol{\omega}_{d}$, where $\boldsymbol{\omega}_{s}\in\mathbb{R}^{3}$ denotes the angular velocity of the body-fixed frame $\mathfrak{R}_{b}$ with respect to the inertia frame $\mathfrak{R}_{i}$ and resolved in $\mathfrak{R}_{b}$; $\boldsymbol{C}_{e}\in\mathbb{R}^{3\times 3}$ represents the transformation matrix from target frame $\mathfrak{R}_{d}$ to the body-fixed frame $\mathfrak{R}_{b}$; $\boldsymbol{\omega}_{d}\in\mathbb{R}^{3}$ denotes the desired angular velocity, resolved in frame $\mathfrak{R}_{d}$.
 $\boldsymbol{J}\in\mathbb{R}^{3\times 3}$ denotes the inertia matrix with respect to frame $\mathfrak{R}_{b}$; $\boldsymbol{u}\in\mathbb{R}^{3}$ stands for the control input; $\boldsymbol{d}\in\mathbb{R}^{3}$ represents the external disturbances; $\boldsymbol{\Omega}_{e}\in\mathbb{R}^{3}$ is an inertia matrix coupling-term, defined as $\boldsymbol{\Omega}_{e} \triangleq \boldsymbol{J}\boldsymbol{\omega}^{\times}_e\boldsymbol{C}_e\boldsymbol{\omega}_d - \boldsymbol{J}\boldsymbol{C}_e\dot{\boldsymbol{\omega}}_d -\boldsymbol{\omega}_s^{\times}\boldsymbol{J}\boldsymbol{\omega}_s$;
$\boldsymbol{F}_{e} \in \mathbb{R}^{3\times 3}$ represents the Jacobian matrix of the kinematics equation, expressed as $\boldsymbol{F}_{e} \triangleq \frac{1}{2}\left[q_{e0}\cdot\boldsymbol{I}_{3} + \boldsymbol{q}^{\times}_{ev}\right]$.

To facilitate the following analysis, we consider the following assumptions:
\begin{assumption}\label{AssJ}
The inertia matrix $\boldsymbol{J}$ is a known positive-definite symmetric matrix.
\end{assumption}
Since parameter uncertainties are not the main focus of this paper, we temporarily ignore this issue. For those circumstances with acceptable parameter uncertainties, they can be treated as a lumped disturbance and then compensated using the observer technique, as stated by existing literature \cite{wang2022finite}.

\begin{assumption}\label{Assd}
The norm of the external disturbance is bounded, and the norm of its differentiation signal, i.e., $\|\dot{\boldsymbol{d}}\|$, is bounded by a known constant $\xi_{d}$ such that $\|\dot{\boldsymbol{d}}\| \le \xi_{d}$ holds for $t\in\left[0,+\infty\right)$.
\end{assumption}
Such an assumption is reasonable since most actual external disturbances in attitude control scenarios are varying with acceptable rates.

\subsection{Precisely-Assigned Performance Constraints}\label{PAPC}

This subsection introduces the concept of Precisely-Assigned Performance constraints (PAP constraints), involving requirements on precise settling time, precise steady-state error, and maximum overshoot elimination. Such performance constraints are mainly constructed based on typical engineering indices used to evaluate the system's time response.

\subsubsection{\textbf{Precise Settling Time}}
Settling time stands for the duration for an error state to transition from its initial status to a designated neighborhood near its equilibrium point. For the $i$-th component of $\boldsymbol{q}_{ev}$, given a constant error bound $\Delta_{e}>0$, the practical settling time of the system, denoted as $T_{s}$, can be considered as a minimized time constant that satisfies the condition: $T_{s}:\left\{|q_{evi}(t)|<\Delta_{e},\, t\geq T_{s}\right\}$.

Following this concept, the context of precise settling time requires the error state to converge and remain within the region: $|q_{evi}(t)|<\Delta_{e}$ just near an assigned time instant $T_{sd}$, which can be specified as:
\begin{equation}\label{PSTcons}
	|T_{s}-T_{sd}| \leq \Delta_{T}
\end{equation}
where $\Delta_{T}>0$ stands for an allowed tiny error margin, given for practical consideration.

\subsubsection{\textbf{Precise Steady-State Error}}
Steady-state error evaluates the control accuracy of the closed-loop system. The context of precise steady-state error requires each error state to be restricted beneath a given upper bound $\Delta_{e}$ after the desired settling time, such that the following relationship holds:
\begin{equation}
	|q_{evi}(t)| < \Delta_{e},\quad\text{for}\quad t\geq T_{sd}
\end{equation}

\subsubsection{\textbf{Maximum Overshoot Elimination}}
Maximum overshoot characterizes the maximum deviation when the error state crosses from the initial side of its equilibrium point to the other. The context of maximum overshoot elimination requires the actual maximum overshoot to be restricted under a given threshold $O_{m}\geq 0$, hence we have:
\begin{equation}
	|\min\left(\text{sgn}(q_{evi}(0))q_{evi}(t)\right)| < O_{m}
\end{equation}

Consequently, it can be observed that PAP constraints characterize a complete ideal transient response, where the system is expected to converge at $t=T_{sd}$, remain satisfying $|q_{evi}(t)|<\Delta_{e}$ for $t\in\left[T_{sd},+\infty\right)$, while appearing with limited overshoot smaller than $O_{m}$.

\subsection{Control Objective}
Based on the concept of the PAP constraint, the main control objective of this paper is to design a control law $\boldsymbol{u}\in\mathbb{R}^{3}$ to guide the system $\boldsymbol{q}_{ev}$, $\boldsymbol{\omega}_{e}$ to a small region near its equilibrium point $\boldsymbol{q}_{ev}=\boldsymbol{0}$, $\boldsymbol{\omega}_{e} = \boldsymbol{0}$ even in the presence of external disturbances $\boldsymbol{d}$. The control law should ensure the satisfaction of PAP constraint for the system's output state, i.e., $\boldsymbol{q}_{ev}$, which involve requirements on precise settling time, precise steady-state error, and maximum overshoot elimination.

\section{Precisely-Assigned Performance Control Scheme}

\subsection{Main Structure and Logic}
In order to address all control objectives involved in the PAP constraints, this section proposes the PAP control scheme. In Subsection \ref{TPC}, we first present the design of the Reference Performance Function (RPF), which characterizes an ideal convergence behavior for each state trajectory. Additionally, we establish a tracking performance constraint, quantifying tube-like state regions that meet the performance requirements. Subsequently, in Subsection \ref{PSC}, by introducing the concept of Control Barrier Function (CBF) \cite{ames2016control}, we derive a sufficient condition for maintaining or achieving the tracking performance constraint, termed the PAP condition. Motivated by the backstepping control method and introducing key concepts from Sontag’s Universal Formula for Stabilization \cite{sontag1989universal}, the PAP controller is designed to guide the system to satisfy the affine constraint derived from the PAP condition. This ensures that each state trajectory will converge and be confined within its corresponding tube-like performance-satisfied region within a finite time duration.
 The main logic of the PAP control scheme is presented in Figure \ref{PAPlogic} for additional explanation.

	\begin{figure}[hbt!]
	\centering 
	\includegraphics[width = 0.52\textwidth]{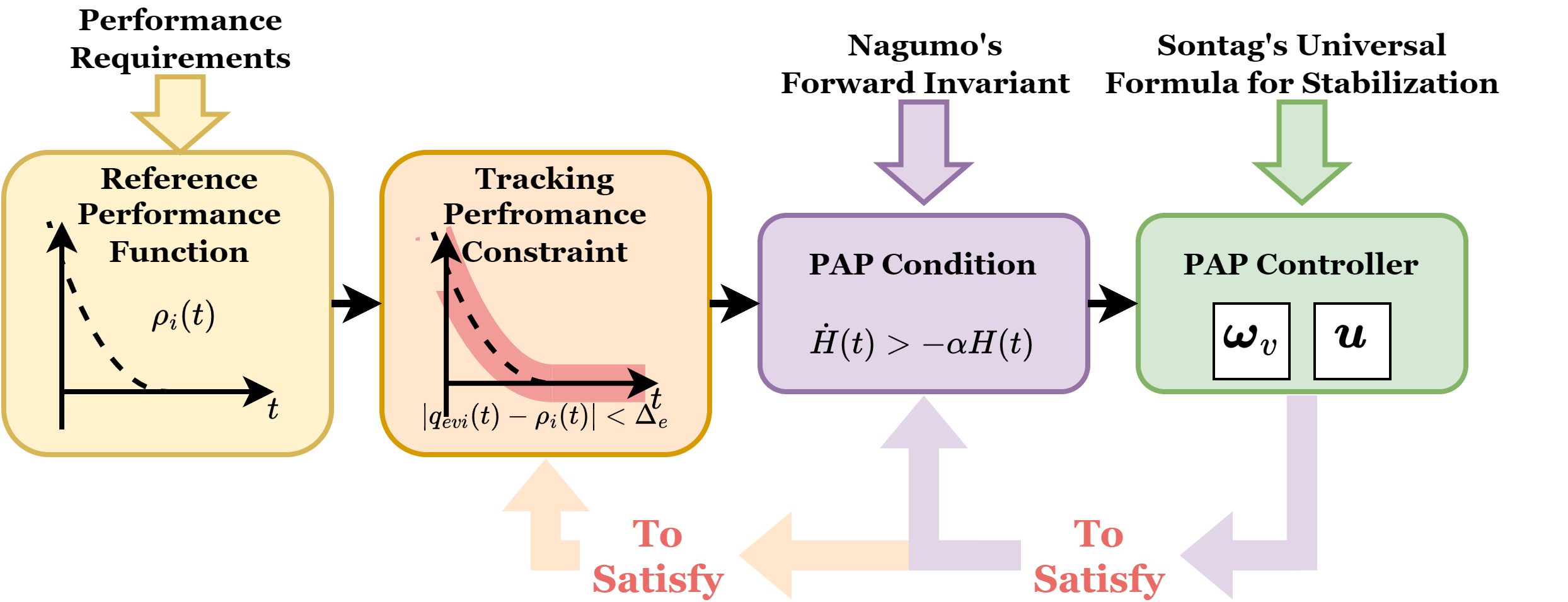}
	\caption{Main Logic of the PAP Control Scheme
	}    
	\label{PAPlogic}  
\end{figure}

\subsection{Tracking Performance Constraint}\label{TPC}

In order to achieve all performance criteria, we depart from the typical funnel-like performance envelopes utilized in Prescribed Performance Control (PPC) schemes. Instead, we guide each error state towards a smooth $\mathcal{C}^{2}$ reference trajectory, specifically designed regarding the given PAP constraints, termed the Reference Performance Function (RPF).

Considering each $i$-th error state of $\boldsymbol{q}_{ev}$ denoted as $q_{evi}(t)(i=1,2,3)$, we denote its corresponding RPF as $\rho_{i}(t)$, constructed based on the following principles:
\begin{equation}
	\begin{aligned}
		\rho_{i}(t) &\ge 0,\ \dot{\rho}_{i}(t) \le 0,\ \rho_{i}(T_{sd}) = 0,\quad\text{if}\quad q_{evi}(0) \ge 0\\
		\rho_{i}(t) &\le 0,\ \dot{\rho}_{i}(t) \ge 0,\ \rho_{i}(T_{sd}) = 0,\quad\text{if}\quad q_{evi}(0) < 0\\
	\end{aligned}
\end{equation}
where $T_{sd} > 0$ stands for the desired settling time instant.
As an example, this paper utilizes the 4th-order polynomial function, as referenced in \cite{hu2017adaptive}, to construct the RPF. The expression of $\rho_{i}(t)$ is given as follows:
\begin{equation}\label{RPFEXPRESSION}
	\begin{aligned}
		\rho_{i}(t) = 
		\begin{cases}
			a_{0i} + \sum_{j=1}^{4}a_{ji}t^{j}, & \quad t \in [0,T_{sd}]\\
			0, & \quad t \in (T_{sd},+\infty)
		\end{cases}
	\end{aligned}
\end{equation}
where $a_{ji}$ for $j=0,1,2,3,4$ are determined by the following smoothness conditions: $\rho_{i}(0) = a_{0i}$, $\dot{\rho}_{i}(0) = 0$, $\rho_{i}(T_{sd}) = 0$, and $\dot{\rho}_{i}(T_{sd}) = 0$.

Based on the defined RPF $\rho_{i}(t)$, we define a tracking performance constraint for each component, specified as follows:
\begin{equation}\label{trackingpercons}
	|q_{evi}(t) - \rho_{i}(t)| < \Delta_{e}
\end{equation}
where $\Delta_{e} > 0$ is a designated constant error bound. It can be observed that the constraint specified by equation  (\ref{trackingpercons}) establishes a tube-like region.
As an example, we use the 4th-order polynomial function to construct the RPF, where parameters are given as $\Delta_{e} = 0.02$, $\rho_{i}(0) = 0.5965$, and $T_{sd} = 80$. As shown in Figure \ref{TrackingRegion}, the red tube-like region represents the state region that satisfies the tracking performance constraint.

\begin{figure}[hbt!]
	\centering 
	\includegraphics[width = 0.52\textwidth]{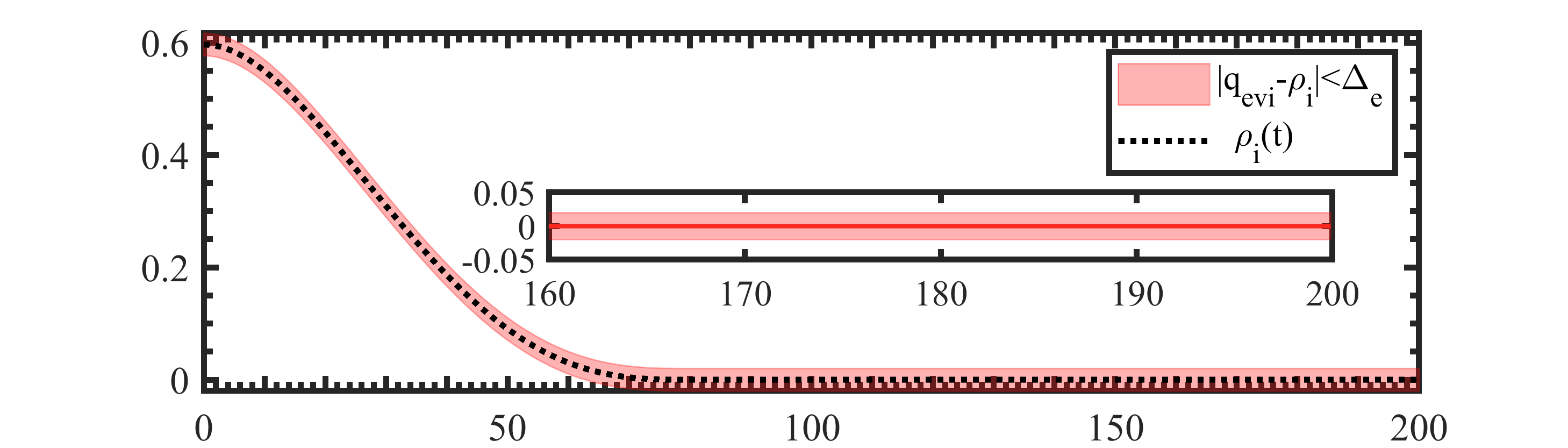}
	\caption{Illustration of the RPF and the state region satisfies the TPC}    
	\label{TrackingRegion}  
\end{figure}

In order to facilitate the following analysis, we then define a tracking error variable as:
\begin{equation}
	\boldsymbol{s} \triangleq \boldsymbol{q}_{ev} - \boldsymbol{\rho}\in\mathbb{R}^{3}
\end{equation}
Therefore, the constraint given by equation (\ref{trackingpercons}) can be reformulated as $|s_{i}(t)| < \Delta_{e}$, where $s_{i}(t)$ represents the $i$-th component of $\boldsymbol{s}$. 

Consequently, with $\rho_{i}(t)\to 0$ and the satisfaction of $|s_{i}(t)|<\Delta_{e}$, each $q_{evi}(t)$ will be restricted in a tube region centered around $\rho_{i}(t)$ (the red region as shown in Figure \ref{TrackingRegion}), and $|q_{evi}(t)|<\Delta_{e}$ will be achieved just at $t=T_{sd}$, remaining for $t>T_{sd}$. By designing a sufficiently small error bound $\Delta_{e}>0$, the PAP constraints can be satisfied.

\subsection{PAP Condition}\label{PSC}
From the end of Subsection \ref{TPC}, our main purpose is to ensure the satisfaction of the tracking performance constraint.
To fulfill such purpose, we migrate the concept of "control to ensure safety" into the realm of "control to satisfy performance requirements", introducing the concept of CBF to derive a sufficient condition for maintaining or achieving the tracking performance constraint as provided by equation (\ref{trackingpercons}), termed the Precisely-Assigned Performance condition (PAP condition).

Before going further, we first introduce the following theorem, as referenced \cite{ames2019control}, known as the Nagumo's forward invariant theorem \cite{nagumo1942lage}:
\begin{lemma}
	\textit{Nagumo's Theorem:}  Given a dynamical system $\dot{\boldsymbol{x}} = \boldsymbol{f}(\boldsymbol{x})$ with $\boldsymbol{x}\in\mathbb{R}^{n}$ the system states, assuming that the safe set $\mathcal{S}$ is defined on the 0-superlevel set of a smooth function (continuously differentiable) $D(\boldsymbol{x})$, i.e, $D(\boldsymbol{x}) = 0$ for $\boldsymbol{x}\in\partial\mathcal{S}$ and $D(\boldsymbol{x})>0$ for $\text{Int}(\mathcal{S})$. Then, the necessary and sufficient condition for set $\mathcal{S}$ to remain invariant is \cite{nagumo1942lage,blanchini1999set}:
	\begin{equation}
		\dot{D}(\boldsymbol{x})\ge0,\quad\text{for}\quad\forall\boldsymbol{x}\in\partial\mathcal{S}
	\end{equation}
\end{lemma}
As introduced in \cite{ames2019control}, such a theorem was used to develop the so-called CBF method, where the authors provide the following minimally restrictive sufficient and necessary conclusion:
\begin{lemma}\label{T1}
	Given a smooth $D(t)$, if the following condition holds:
	\begin{equation}
		\begin{aligned}
			\dot{D}(t) \ge -\alpha_{0}\left(D(t)\right), \quad D(0) > 0
		\end{aligned}
	\end{equation}
	then $D(t) > 0$ holds for $t\in\left[0,+\infty\right)$, and the set defined by $\mathcal{D}: \left\{D(t)>0\right\}$ is forward-invariant.
	In the above expression, $\alpha_{0}(\cdot): \mathbb{R} \to \mathbb{R}$ stands for an extended $\mathcal{K}_{\infty}$ function that is continuous and strictly monotonically increasing, satisfying $\alpha_{0}(0) = 0$, $\lim_{r \to +\infty} \alpha_{0}(r) = +\infty$, and $\lim_{r \to -\infty} \alpha_{0}(r) = -\infty$.
\end{lemma}

\begin{remark}
The above theorem is presented considering a general extended $\mathcal{K}_{\infty}$ function $\alpha_{0}(\cdot)$. However, one can conveniently choose a positive constant $\alpha > 0$ as a special case. Thus, this simplifies the condition given in Theorem \ref{T1} to: $\dot{D}(t) > -\alpha D(t)$, such a special case can be also seen from \cite{aubin2011viability}.
\end{remark}

We define a scalar function $H(t)$, given as follows:
\begin{equation}\label{Hfunction}
H(t) \triangleq K_{H}\left(\Delta^{2}_{e} - \sum_{i=1}^{3}s^{2}_{i}(t)\right)=K_{H}\left(\Delta^{2}_{e}-\boldsymbol{s}^{\text{T}}\boldsymbol{s}\right)
\end{equation}
where $K_{H}>0$ is a positive constant gain parameter; $\Delta_{e}>0$ is a positive threshold parameter.

From the definition of $H(t)$, it is not hard to obtain that $H(t)>0$ is a strengthened sufficient condition for the satisfaction of $|s_{i}(t)|<\Delta_{e},\forall i=1,2,3$. We define a set $\mathcal{H}_{s}:\left\{H(t)>0\right\}$ and a mutually exclusive set $\mathcal{H}_{u}:\left\{H(t)\le0\right\}$, called the performance-satisfied set and the performance-unsatisfied set, respectively. Since the $0$-level set of $H(t)$ separates the state regions $\mathcal{H}_{s}$ and $\mathcal{H}_{u}$, despite some conservativeness, $H(t)$ can be used to indicate the satisfaction of the tracking performance constraint.

We then discuss under which condition, the system is able to maintain, or to recover to satisfy $H(t)>0$. Applying on the conclusion of Theorem \ref{T1} and the property of exponential functions, we have the following three derivations:
\begin{itemize}
	\item[1)] 
	For $\dot{H}(t) + \alpha H(t) \ge 0$ and $H(0) > 0$, $H(t) > 0$ always holds, the system is performance-satisfied, and the set $\mathcal{H}_{s}$ is forward-invariant.
	
		\item[2)] 	For $\dot{H}(t) + \alpha H(t) < 0$ and $H(0) > 0$, although the system is currently performance-satisfied, since $H(t)$ has the possibility of being $H(t) \le 0$, the system suffers from the risk of being performance-unsatisfied.
		
			\item[3)] 	For $\dot{H}(t) + \alpha H(t) > 0$ but with $H(t) \le 0$, this shows that the system is currently performance-unsatisfied. However, it has the trend to recover to be performance-satisfied again.
\end{itemize}
\begin{figure}[hbt!]
	\centering 
	\includegraphics[width = 0.4\textwidth]{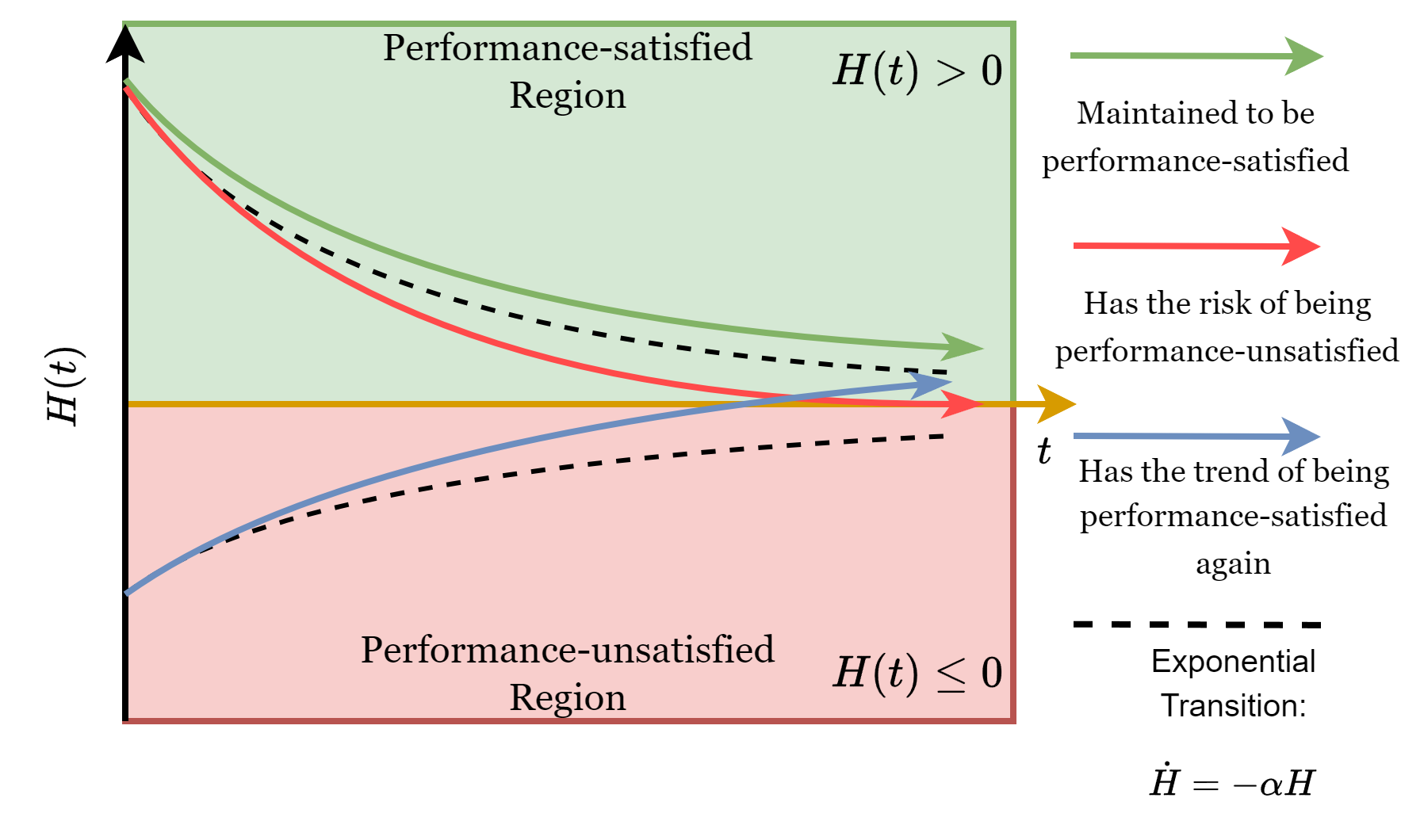}
	\caption{Illustration Explanation of the above Conclusions}    
	\label{PAPcondition}  
\end{figure}

A simple and intuitive explanation of these conclusions is illustrated in Figure \ref{PAPcondition}: if $\dot{H}(t) > -\alpha H(t)$ holds, then according to the Comparison lemma \cite{khalil2002control}, the actual time-response of $H(t)$ will be above the one given by the ODE: $\dot{y}=-\alpha y, y(0)=H(0)$. Thus, since $\lim_{t\to+\infty}y(t)\to0$ holds for both $H(0)\ge 0$ and $H(0)<0$, this leads to the three conclusions mentioned earlier. In Figure \ref{PAPcondition}, the black-dotted lines represent exponential-transition responses for different initial conditions $H(0)\ge 0$ and $H(0) < 0$.

From the above discussion, note that if the system satisfies $\dot{H}(t) + \alpha H(t) > 0$ for $H(0) \leq 0$, or satisfies $\dot{H}(t) + \alpha H(t) \geq 0$ for $H(0) > 0$, the system is able to \textbf{maintain} (condition (1)) or \textbf{recover} (condition (3)) to the performance-satisfied region $\mathcal{H}_{s}$. Therefore, we can consider a strengthened condition for all possible situations by removing the equality condition, which gives rise to the so-called \textit{PAP Condition}:

	\begin{equation}\label{PAPConditionexpression}
		\dot{H}(t)>-\alpha H(t)
	\end{equation}
	While such a condition ensures that $H(t)>0$ will be ultimately achieved under any circumstances, it can be equivalently reformulated by considering the time-derivative of $\boldsymbol{s}$, expressed as:
	\begin{equation}\label{affineconstraint}
		-2K_{H}\boldsymbol{s}^{\text{T}}\left(\boldsymbol{F}_{e}\boldsymbol{\omega}_{e}-\dot{\boldsymbol{\rho}}\right)+\alpha H(t)>0
	\end{equation}
	
	 Consequently, our next key step is to guide the system to satisfy the affine constraint shown in equation (\ref{PAPConditionexpression}), or equivalently, the constraint in equation (\ref{affineconstraint}).

\subsection{PAP Controller}\label{VTC}
Motivated by the Backstepping control method, this subsection proposes a Precisely-Assigned Performance (PAP) controller to guide the system in satisfying the affine constraint as shown in equation (\ref{affineconstraint}). The PAP controller consists of a virtual control law $\boldsymbol{\omega}_{v}\in\mathbb{R}^{3}$ and an actual tracking control law $\boldsymbol{u}\in\mathbb{R}^{3}$, developed by incorporating the main ideas from Sontag's Universal Formula for Stabilization, as originally proposed in \cite{sontag1989universal}. Additionally, to address the perturbation issue, we simultaneously employ the extended state observer technique to construct a disturbance observer.

\subsubsection{The Design of the Virtual Control Law} 

Firstly, we present the design of the smooth virtual control law $\boldsymbol{\omega}_{v}\in\mathbb{R}^{3}$, expressed as follows:
\begin{equation}\label{omegav}
	\boldsymbol{\omega}_{v} = \boldsymbol{F}_{e}^{-1}\left[\left(2\lambda_{v}(A_{1},B_{1})K_{H}-K_{s}\right)\boldsymbol{s}+\dot{\boldsymbol{\rho}}\right]
\end{equation}
where $K_{s}>0$ is a positive constant gain parameter; $\lambda_{v}\left(A_{1},B_{1}\right):\mathbb{R}\times\mathbb{R}\to\mathbb{R}$ is a scalar function, given as:
\begin{equation}\label{lambdav}
	\lambda_{v}\left(A_{1},B_{1}\right) \triangleq 
	\begin{cases}
		0,\quad &B_{1}=0\\
		\frac{-A_{1}-\sqrt{A_{1}^{2}+q_{1}(B_{1})B_{1}}}{B_{1}},\quad&B_{1}\neq0
	\end{cases} 
\end{equation}
where $q_{1}:\mathbb{R}\to\mathbb{R}$ is a smooth function that satisfies $q_{1}(0)=0$ and $q_{1}(B_{1})>0$ for $\forall B_{1}\neq 0$, chosen as $q_{1}(B_{1}) = \sigma_{1} B_{1}$ with $\sigma_{1}>0
$ a positive constant; $A_{1}$, $B_{1}$ are two scalar variables given as follows:
\begin{equation}\label{A1B1}
	\begin{aligned}
		A_{1} &=  -\alpha H(t) + \delta_{H}\|\tanh(C_{s}\boldsymbol{s})\|\\
		B_{1} &= \|2K_{H}\boldsymbol{s}\|^{2} = 4K^{2}_{H}\boldsymbol{s}^{\text{T}}\boldsymbol{s}
	\end{aligned}
\end{equation}
where $\delta_{H}, C_{s}>0$ are positive constant parameters. For writing convenience, $\lambda_{v}(A_{1},B_{1})$ is simplified as $\lambda_{v}$ for the following analysis.

\subsubsection{The Design of the Actual Control Law} 
Subsequently, we present the design of the actual control law $\boldsymbol{u}\in\mathbb{R}^{3}$. Before process further, we define a tracking error variable $\boldsymbol{z}_{2}\in\mathbb{R}^{3}$ as: $\boldsymbol{z}_{2} \triangleq \boldsymbol{\omega}_{e} - \boldsymbol{\omega}_{v}$.
We then define a function $h(t)$ as follows:
\begin{equation}
	h(t) \triangleq K_{h}\left(\Delta^{2}_{h}-\boldsymbol{z}^{\text{T}}_{2}\boldsymbol{z}_{2}\right)
\end{equation}
where $K_{h}>0$ is a positive constant gain parameter, $\Delta_{h}>0$ is a positive threshold parameter.
Accordingly, $\boldsymbol{u}$ is designed as follows:
\begin{equation}\label{u}
	\begin{aligned}
		\boldsymbol{u} &= -\boldsymbol{\Omega}_{e} - \hat{\boldsymbol{d}} + \boldsymbol{J}\dot{\boldsymbol{\omega}}_{v}
		 + \left[2\lambda_{u}(A_{2},B_{2})K_{h}-K_{2}\right]\boldsymbol{J}\boldsymbol{z}_{2}
	\end{aligned}
\end{equation}
where $K_{2}>0$ is a positive constant gain parameter; $\boldsymbol{\Omega}_{e}\in\mathbb{R}^{3}$ represents the dynamical-coupling term; $\hat{\boldsymbol{d}}\in\mathbb{R}^{3}$ stands for the estimation of the observer, later given in Subsection \ref{observer}; $\dot{\boldsymbol{\omega}}_{v}\in\mathbb{R}^{3}$ stands for the time-derivative of $\boldsymbol{\omega}_{v}$; $\lambda_{u}(A_{2},B_{2}):\mathbb{R}\times\mathbb{R}\to\mathbb{R}$ is another scalar function, expressed as:
\begin{equation}
	\lambda_{u}\left(A_{2},B_{2}\right) \triangleq 
	\begin{cases}
		0,\quad &B_{2}=0\\
		\frac{-A_{2}-\sqrt{A_{2}^{2}+q_{2}(B_{2})B_{2}}}{B_{2}},\quad&B_{2}\neq0
	\end{cases} 
\end{equation}
Similarly, $q_{2}(B_{2})$ is selected as $q_{2}(B_{2}) = \sigma_{2}B_{2}$, where $\sigma_{2}>0$ is a positive constant; $A_{2}\in\mathbb{R}$ and $B_{2}\in\mathbb{R}$ are given as:
\begin{equation}\label{A2B2}
	\begin{aligned}
		A_{2} &=  -\lambda^{2}_{J\min}\left(\gamma h(t) + \delta_{h}\|\boldsymbol{z}_{2}\|\right)\\
		B_{2} &= \|2K_{h}\boldsymbol{J}\boldsymbol{z}_{2}\|^{2} = 4K^{2}_{h}\boldsymbol{z}_{2}^{\text{T}}\boldsymbol{J}^{\textbf{T}}\boldsymbol{J}\boldsymbol{z}_{2}
	\end{aligned}
\end{equation}
where both $\delta_{h}>0$ and $\gamma>0$ are positive constants; $\lambda_{J\min}$ stands for the minimum eigenvalue of $\boldsymbol{J}$. Similarly, $\lambda_{u}(A_{2},B_{2})$ is written as $\lambda_{u}$ for brevity. Notably, since $\boldsymbol{J} = \boldsymbol{J}^{\text{T}}$ holds, then we have $B_{2} = 4K^{2}_{h}\boldsymbol{z}^{\text{T}}_{2}\boldsymbol{J}^{2}\boldsymbol{z}_{2}$.

To facilitate the following analysis, we discuss about the expression of $\dot{\boldsymbol{\omega}}_{v}$, which can be analytically given as:
\begin{equation}\label{domegav}
	\begin{aligned}
		\dot{\boldsymbol{\omega}}_{v} &= \dot{\boldsymbol{F}_{e}^{-1}}\left[\left(2\lambda_{v}K_{H}-K_{s}\right)\boldsymbol{s}+\dot{\boldsymbol{\rho}}\right]\\
		&\quad + \boldsymbol{F}^{-1}_{e}\left[2K_{H}\left(\dot{\lambda_{v}}\boldsymbol{s}+\lambda_{v}\dot{\boldsymbol{s}}\right)-K_{s}\dot{\boldsymbol{s}}+\ddot{\boldsymbol{\rho}}\right]
	\end{aligned}
\end{equation}
where $\dot{\lambda}_{v}$ stands for the time-derivative of $\lambda_{v}$. According to the expression of $\lambda_{v}$ as shown in equation (\ref{lambdav}), since both $-\alpha H(t)$ and $\delta_{H}\|\tanh(C_{s}\boldsymbol{s})\|$ are smooth functions due to the smoothness of $\boldsymbol{s}$, by applying the theorem of implicit function, the time-derivative of $\lambda_{v}$ exists and can be then given as:
\begin{equation}
	\dot{\lambda}_{v} = \frac{\partial\lambda_{v}}{\partial A_{1}}\dot{A}_{1}(t) + \frac{\partial\lambda_{v}}{\partial B_{1}}\dot{B}_{1}(t)
\end{equation}
where $\dot{A}_{1}(t) = \left(-\alpha\dot{H}(t)
+
\frac{\delta_{H}}{\|\tanh(C_{s}\boldsymbol{s})\|}(\tanh(C_{s}\boldsymbol{s}))^{\text{T}}\boldsymbol{M}_{s}\dot{\boldsymbol{s}}\right)$ holds with $\boldsymbol{M}_{s} = \text{diag}\left(C_{s}\left(1-\tanh^{2}(C_{s}s_{i}(t))\right)\right)\in\mathbb{R}^{3\times 3}$ a diagonal matrix, and $\dot{B}_{1}(t)$ is given as $\dot{B}_{1}(t) = 8K^{2}_{H}\boldsymbol{s}^{\text{T}}\dot{\boldsymbol{s}}$;  
$\dot{\boldsymbol{s}}$ is previously given as $\dot{\boldsymbol{s}} = \boldsymbol{F}_{e}\boldsymbol{\omega}_{e} - \dot{\boldsymbol{\rho}}$, and $\dot{\boldsymbol{F}_{e}^{-1}} = -\boldsymbol{F}^{-1}_{e}\dot{\boldsymbol{F}}_{e}\boldsymbol{F}^{-1}_{e}$ holds with $\dot{\boldsymbol{F}}_{e} = \frac{1}{2}\left(-\frac{1}{2}\boldsymbol{q}^{\text{T}}_{ev}\boldsymbol{\omega}_{e}\boldsymbol{I}_{3}+\left(\boldsymbol{F}_{e}\boldsymbol{\omega}_{e}\right)^{\times}\right)$. All aforementioned parts are calculable, and according to the theorem \cite{sontag1989universal}, both $\lambda_{u}$ and $\lambda_{v}$ are smooth functions outside zeroing point $B_{1}=0$ and $B_{2}=0$, also mentioned in \cite{cohen2023characterizing}.
\begin{remark}
Note that when $B_{1}$ and $B_{2}$ approaches $B_{1}=0$, $B_{2} = 0$, $\lambda_{u}$ and $\lambda_{v}$ has the possibility of increasing intensively to $+\infty$. To practically avoid such a drawback, one can add a small term $\epsilon>0$ on the denominator as: $B_{1}+\epsilon$ and $B_{2}+\epsilon$. Such a treatment is used in our simulation section, where we set $\epsilon = 1\times 10^{-7}$.
\end{remark}

\subsubsection{Disturbance Observer Design}\label{observer}
In this subsection, we employ the extended-state observer technique, originally proposed in \cite{li2016extended}, to derive a disturbance observer.

Let $\boldsymbol{F}_{1}\triangleq \boldsymbol{\omega}_{e}\in\mathbb{R}^{3}$, and we regard $\boldsymbol{J}^{-1}\boldsymbol{d}$ as another extended state of the observer, denoted as $\boldsymbol{F}_{2} \triangleq \boldsymbol{J}^{-1}\boldsymbol{d}$. Define $\boldsymbol{\xi} \triangleq \boldsymbol{J}^{-1}\dot{\boldsymbol{d}}$. Considering the dynamics equation of $\boldsymbol{\omega}_{e}$ as shown in equation (\ref{errsystem}), this generates the following system:
\begin{equation}
	\begin{aligned}
		\dot{\boldsymbol{F}}_{1} &= \boldsymbol{F}_{d}+\boldsymbol{J}^{-1}\boldsymbol{u}+\boldsymbol{F}_{2}, \quad\dot{\boldsymbol{F}}_{2} = \boldsymbol{\xi}
	\end{aligned}
\end{equation}
where $\boldsymbol{F}_{d}\in\mathbb{R}^{3}$ can be expressed as $\boldsymbol{F}_{d} = \boldsymbol{J}^{-1}\boldsymbol{\Omega}_{e}$. Defining the estimation of $\boldsymbol{F}_{1}$ and $\boldsymbol{F}_{2}$ as $\hat{\boldsymbol{F}}_{1}$ and $\hat{\boldsymbol{F}}_{2}$ respectively, the estimation error can be defined as $\boldsymbol{e}_{1} \triangleq \hat{\boldsymbol{F}}_{1} - \boldsymbol{F}_{1}$ and $\boldsymbol{e}_{2}\triangleq \hat{\boldsymbol{F}}_{2} - \boldsymbol{F}_{2}$. The observer design can then be specified as:
\begin{equation}\label{observerdesign}
	\begin{aligned}
		\dot{\hat{\boldsymbol{F}}}_{1} &= \boldsymbol{F}_{d} + \boldsymbol{J}^{-1}\boldsymbol{u} + \hat{\boldsymbol{F}}_{2} - C_{1}\beta\boldsymbol{e}_{1}, \\
		\dot{\hat{\boldsymbol{F}}}_{2} &= -C_{2}\beta^{2}\boldsymbol{e}_{1}
	\end{aligned}
\end{equation}
where $C_{1}>0$ and $C_{2}>0$ are two observer gain constants, and $\beta>0$ is a design parameter. Based on the estimate of $\boldsymbol{F}_{2}$, the estimate of $\boldsymbol{d}$ can be constructed as $\hat{\boldsymbol{d}} = \boldsymbol{J}\hat{\boldsymbol{F}}_{2}$.

The convergence analysis of the observer design will be given later in Appendix \ref{append}.

\begin{remark}

It should be noted that in \cite{verginis2022funnel}, the control law is derived directly using the CBF-QP technique, which necessitates solving an optimization problem for controller implementation. Our control scheme, inspired primarily by the backstepping method, features a hierarchical structure. Here, we design a virtual control signal (velocity) and then track it to satisfy the required constraints. However, using the CBF-QP directly does not ensure a smooth virtual control signal since it only guarantees Lipschitz continuity. This limitation motivates the introduction of Sontag's universal formula for stabilization, a key element in our design that enables the construction of a smooth virtual control signal. Consequently, this allows for the direct use of its time derivative, i.e., $\dot{\boldsymbol{\omega}}_{v}$, to tightly track the virtual control signal as required by the backstepping control method. Furthermore, our proposed PAP control scheme avoids the need to solve an optimization problem, resulting in a control structure similar to that of backstepping control with time-varying gains determined by $\lambda_{v}(A_{1}, B_{1})$ and $\lambda_{u}(A_{2}, B_{2})$.

%Another notable difference is that in \cite{verginis2022funnel}, the bounds of the performance envelope are considered obstacles that exert a repulsion effect, thereby constraining error states to remain within them. In our control scheme, however, the control objective is to lead the system to converge to a tube, viewing the problem as a stabilization task. Thus, we employ the perspective of a Control Lyapunov Function and subsequently use Sontag's universal formula for stabilization to derive an analytical control law.

\end{remark}

\section{Stability Analysis}
The stability analysis of the proposed controller mainly involves two perspectives. In Subsection \ref{basicstable}, we demonstrate that the proposed controller ensures the Uniformly Ultimately Boundedness (UUB) of the overall closed-loop system. Subsequently, in Subsection \ref{Attract}, we focus on proving that the proposed PAP controller effectively guides all state trajectories into the performance-satisfying tube regions, ensuring that $H(t) > 0$ is achieved within a finite time duration.

\subsection{Boundedness of the Closed-Loop System}\label{basicstable}
In this subsection, we first demonstrate that the closed-loop system is UUB with the proposed PAP controller under any circumstance. The main theorem is presented as follows:

\begin{theorem}\label{Tmain1}
For the attitude error system given by equation (\ref{errsystem}), under Assumptions \ref{AssJ}, \ref{Assd}, with the PAP controller consisting of the virtual control law $\boldsymbol{\omega}_{v}\in\mathbb{R}^{3}$, the actual control law $\boldsymbol{u}\in\mathbb{R}^{3}$, and the disturbance observer design shown in equation (\ref{observerdesign}), the closed-loop system is UUB.
\end{theorem}
\begin{proof}
	Choosing a candidate Lyapunov function as $V_{s} = \frac{1}{2}K_{H}\boldsymbol{s}^{\text{T}}\boldsymbol{s}$, taking its time-derivative and combines with the design of $\boldsymbol{\omega}_{v}$ from equation (\ref{omegav}), it can be obtained that:
	\begin{equation}\label{dotVs1}
		\begin{aligned}
			\dot{V}_{s}  &= K_{H}\boldsymbol{s}^{\text{T}}\left(\boldsymbol{F}_{e}\boldsymbol{\omega}_{v}-\dot{\boldsymbol{\rho}}\right) + K_{H}\boldsymbol{s}^{\text{T}}\boldsymbol{F}_{e}\boldsymbol{z}_{2} \\
			&= 2K^{2}_{H}\lambda_{v}\|\boldsymbol{s}\|^{2} - K_{H}K_{s}\|\boldsymbol{s}\|^{2} + K_{H}\boldsymbol{s}^{\text{T}}\boldsymbol{F}_{e}\boldsymbol{z}_{2} \\
		\end{aligned}
	\end{equation}
	Given that $\lambda_{v} > 0$ when $B_{1} \neq 0$ holds, owing to the fact that $\|\boldsymbol{F}_{e}\|=\frac{1}{2}$, we derive:
	\begin{equation}\label{dotVs}
		\begin{aligned}
			\dot{V}_{s} &\le -K_{H}\left(K_{s} - \frac{a_{1}}{4}\right)\|\boldsymbol{s}\|^{2} + \frac{K_{H}}{4a_{1}}\|\boldsymbol{z}_{2}\|^{2}
		\end{aligned}
	\end{equation}
	
	For the tracking error system defined by $\boldsymbol{z}_{2}$, we select a candidate Lyapunov function $V_{z} = \frac{1}{2}\boldsymbol{z}_{2}^{\text{T}}\boldsymbol{J}\boldsymbol{z}_{2}$. Taking the time derivative of $V_{z}$ and combines it with the control law $\boldsymbol{u} \in \mathbb{R}^{3}$ from equation (\ref{u}), we yield:
	\begin{equation}
		\begin{aligned}
			\dot{V}_{z} &= \left(2\lambda_{u}K_{h} - K_{2}\right)\boldsymbol{z}_{2}^{\text{T}}\boldsymbol{J}\boldsymbol{z}_{2} + \boldsymbol{z}_{2}^{\text{T}}\tilde{\boldsymbol{d}} \\
			&\le -K_{2}\boldsymbol{z}_{2}^{\text{T}}\boldsymbol{J}\boldsymbol{z}_{2} + \frac{b_{1}}{2}\|\boldsymbol{z}_{2}\|^{2} + \frac{1}{2b_{1}}\|\tilde{\boldsymbol{d}}\|^{2}
		\end{aligned}
	\end{equation}
	where $b_{1} > 0$ is a positive constant; $\tilde{\boldsymbol{d}} = \boldsymbol{d} - \hat{\boldsymbol{d}} \in \mathbb{R}^{3}$ represents the estimation error of lumped disturbances. As shown in Appendix \ref{append}, there exists a sufficiently large positive constant $\xi_{m}$ such that $\|\tilde{\boldsymbol{d}}\| \le \xi_{m}$ holds. Meanwhile, note that $\|\boldsymbol{z}_{2}\|^{2} \le \frac{2V_{z}}{\lambda_{J\min}}$ holds, this gives a rise to:
	\begin{equation}\label{dotVz}
		\dot{V}_{z} \le -\left(2K_{2} - \frac{b_{1}}{\lambda_{J\min}}\right)V_{z} + \frac{\xi^{2}_{m}}{2b_{1}}
	\end{equation}
Combining the results in equations \eqref{dotVs}, \eqref{dotVz}, and Appendix~\ref{append}, we can deduce that:
\begin{equation}\label{dotVall1}
	\begin{aligned}
		&\dot{V}_{s}+\dot{V}_{z}+\dot{V}_{\varepsilon} \\
		\le& -2\left(K_{s}-\frac{a_{1}}{4}\right)V_{s}-\left(2K_{2}-\frac{2a_{1}b_{1}+K_{H}}{2a_{1}\lambda_{J\min}}\right)V_{z}\\
		&\quad -C_{d}V_{\varepsilon}+\frac{\xi^{2}_{m}}{2b_{1}}+h_{m}
	\end{aligned}
\end{equation}
Further considering a lumped Lyapunov certificate $V$ as $V \triangleq V_{s}+V_{z}+V_{\varepsilon}$, from equation \eqref{dotVall1}, it can be derived that:
\begin{equation}\label{dotVall}
	\dot{V}\le -K_{V}V+h_{m}+\frac{\xi^{2}_{m}}{2b_{1}}
\end{equation}
where $K_{V}$ is a positive constant expressed as $K_{V}\triangleq \min\left[2\left(K_{s}-\frac{a_{1}}{4}\right),\left(2K_{2}-\frac{2a_{1}b_{1}+K_{H}}{2a_{1}\lambda_{J\min}}\right),C_{d}\right]$. Notably, we should ensure that $K_{s}-\frac{a_{1}}{4}>0$, $2K_{2}-\frac{2a_{1}b_{1}+K_{H}}{2a_{1}\lambda_{J\min}}>0$, and $C_{d}>0$ hold, meaning that $K_{s}$ and $K_{2}$ should be chosen to be big enough.

Since equation~\eqref{dotVall} can be rewritten as: $\dot{V}(t) \le -(1-\mu)K_{V}V(t)-\mu K_{V}V(t) + h_{m}+\xi^{2}_{m}/2b_{1}$, where $\mu\in\left(0,1\right)$ is a constant. Thus, if $\mu K_{V} V(t)>h_{m}+\frac{\xi^{2}_{m}}{2b_{1}}$ holds, then we have $\dot{V}(t)\le -(1-\mu)K_{V}V(t)$.
This suggests that the closed-loop system will exponentially converge to a compact residual set, thereby proving that the closed-loop system is UUB. This completes the proof of Theorem \ref{Tmain1}.

\end{proof}

\subsection{Attractivity Analysis of the Tube Region}\label{Attract}

In this subsection, we demonstrate that the proposed PAP controller effectively guides the system to satisfy $H(t)>0$ within a limited time duration under all circumstances. This indicates that each state trajectory will finally converge and be confined within its corresponding "tube" region, thereby satisfying every performance criteria involved in the context of the PAP constraints.
 
 Main theorems are stated as follows:
\begin{theorem}\label{T2}
	
\begin{enumerate}
	\item For $h(0) > 0$ and $H(0) > 0$, the system satisfies $H(t) > 0$ for $t \in [0,+\infty)$. Therefore, this shows that the tracking performance constraint will be always adhered.
	
	\item For $h(0) > 0$ but with $H(0) \leq 0$, the system is initially performance-unsatisfied. However, it will take at most $T_{H1}$ time to lead the system back to the status $H(t) > 0$, and hence the system will be performance-satisfied again for $t \in [T_{H1},+\infty)$, where $T_{H1}$ is given as:
	\begin{equation}\label{TH1}
		T_{H1} = \frac{1}{\alpha}\ln\left(\frac{\alpha|H(0)|}{\delta_{S}\Delta_{e}} + 1\right)
	\end{equation}
	where $\alpha$, $\delta_{S}$, and $\Delta_{e}$ are design parameters.
	
	\item For $h(t) \leq 0$, under the worst-case scenario, the system will be performance-satisfied for $t \in [T_{h}+T_{H2},+\infty)$, where $T_{h}$ and $T_{H2}$ are given as:
	\begin{equation}\label{Th}
		\begin{aligned}
			T_{h} &= \frac{1}{\gamma}\ln\left(\frac{\gamma|h(0)|}{\delta_{z}\Delta_{h}} + 1\right)\\
			T_{H2} &= \frac{1}{\alpha}\ln\left(\frac{\alpha|\mathcal{H}_{B}|}{\delta_{S}\Delta_{e}} + 1\right)
		\end{aligned}
	\end{equation}
	where $\mathcal{H}_{B}$ is defined as $\mathcal{H}_{B} \triangleq H(0)e^{-\alpha T_{h}} + \frac{\delta_{S}\Delta_{e}}{\alpha}(1-e^{-\alpha T_{h}}) - \mathcal{G}_{B}$, and $\mathcal{G}_{B}$ is given in equation (\ref{GB}).
\end{enumerate}

\end{theorem}

\textit{Proof:} \subsubsection{Attractivity Anslysis of the Tube Region $h(t)>0$}\label{attracth}

We start by considering the time response of $h(t)$, showing that the system will be guided to satisfy $h(t) > 0$ within a limited time duration. The function $h(t)$ has been previously defined as $h(t) = K_{h}\left(\Delta^{2}_{h} - \boldsymbol{z}^{\text{T}}_{2}\boldsymbol{z}_{2}\right)$.

Taking the time-derivative of $h(t)$ and combining it with the dynamics of $\dot{\boldsymbol{\omega}}_{e}$ as shown in equation (\ref{errsystem}), we obtain:
\begin{equation}
	\dot{h}(t) = -2K_{h}\boldsymbol{z}^{\text{T}}_{2}\boldsymbol{J}^{-1}\left[\boldsymbol{\Omega}_{e} + \boldsymbol{u} + \boldsymbol{d} - \boldsymbol{J}\dot{\boldsymbol{\omega}}_{v}\right]
\end{equation}
Substituting the design of $\boldsymbol{u}$ given by equation (\ref{u}) into the expression, this gives rise to:
\begin{equation}\label{doth1}
	\begin{aligned}
		\dot{h}(t) &= -2K_{h}\boldsymbol{z}^{\text{T}}_{2}\boldsymbol{J}^{-1}\tilde{\boldsymbol{d}} - 2K_{h}(2\lambda_{u}K_{h} - K_{2})\boldsymbol{z}^{\text{T}}_{2}\boldsymbol{z}_{2}
	\end{aligned}
\end{equation}
where $\tilde{\boldsymbol{d}}\in\mathbb{R}^{3}$ has been previously defined as $\tilde{\boldsymbol{d}} = \boldsymbol{d} - \hat{\boldsymbol{d}}$, representing the estimation error of perturbations. Thus, according to the analysis presented in Subsection \ref{basicstable}, for the term expressed as $-2K_{h}\boldsymbol{z}^{\text{T}}_{2}\boldsymbol{J}^{-1}\tilde{\boldsymbol{d}}$, we have $-2K_{h}\boldsymbol{z}^{\text{T}}_{2}\boldsymbol{J}^{-1}\tilde{\boldsymbol{d}} \ge -2\frac{K_{h}\xi_{m}}{\lambda_{J\min}}\|\boldsymbol{z}_{2}\|$.

Substituting this result into equation (\ref{doth1}), it can be further derived that:
\begin{equation}\label{doth2}
	\begin{aligned}
		\dot{h}(t) &\ge -2\frac{K_{h}\xi_{m}}{\lambda_{J\min}}\|\boldsymbol{z}_{2}\|-4K^{2}_{h}\lambda_{u}\boldsymbol{z}^{\text{T}}_{2}\boldsymbol{z}_{2}+2K_{h}K_{2}\boldsymbol{z}^{\text{T}}_{2}\boldsymbol{z}_{2}\\
	\end{aligned}
\end{equation}
Since $\lambda_{u}>0$ holds for $B_{2}\neq 0$, while $\lambda_{u} = 0$ holds for $B_{2} = 0$, then we have $-4K^{2}_{h}\lambda_{u}\boldsymbol{z}^{\text{T}}_{2}\boldsymbol{z}_{2} \ge -\frac{4K^{2}_{h}}{\lambda^{2}_{J\min}}\lambda_{u}\left(\boldsymbol{J}\boldsymbol{z}_{2}\right)^{\text{T}}\left(\boldsymbol{J}\boldsymbol{z}_{2}\right)$.

Notably, the expression for $\lambda_{u}$ takes different forms depending on the value of $B_{2}$.
In the special case where $B_{2} = 0$, this implies $h(t) = K_{h}\Delta^{2}_{h} > 0$. From equation (\ref{doth2}), we obtain $\dot{h}(t) \ge 0 > -\gamma h(t)$. This indicates that the time response of $\boldsymbol{z}_{2}$ has been already confined within the tube: $\boldsymbol{z}_{2}:\left\{\|\boldsymbol{z}_{2}\| < \Delta_{h}\right\}$, and will continue to satisfy, as per derived PAP condition.
Thus, our main discussion will focus on the case where $B_{2} \neq 0$:

Since $B_{2} = \|2K_{h}\boldsymbol{J}\boldsymbol{z}_{2}\|^{2} = 4K^{2}_{h}\left(\boldsymbol{J}\boldsymbol{z}_{2}\right)^{\text{T}}\boldsymbol{J}\boldsymbol{z}_{2}$, we have:
	\begin{equation}
		\dot{h}(t) \ge -2\frac{K_{h}\xi_{m}}{\lambda_{J\min}}\|\boldsymbol{z}_{2}\|-\frac{1}{\lambda^{2}_{J\min}}B_{2}\lambda_{u}+2K_{2}K_{h}\boldsymbol{z}^{\text{T}}_{2}\boldsymbol{z}_{2}\\
	\end{equation}
	Note that $-B_{2}\lambda_{u} = A_{2}+\sqrt{A^{2}_{2}+q_{2}(B_{2})B_{2}} > A_{2}$, while $2K_{2}K_{h}\|\boldsymbol{z}_{2}\|^{2} \ge 0$ always holds. Recalling the definition of $A_{2}$ as shown in equation (\ref{A2B2}), we further derive:
	\begin{equation}
		\dot{h}(t) >  -\gamma h(t) + \left(\delta_{h}-2\frac{K_{h}\xi_{m}}{\lambda_{J\min}}\right)\|\boldsymbol{z}_{2}\|
	\end{equation}
	Let $\delta_{h} = \delta_{z}+\frac{2K_{h}\xi_{m}}{\lambda_{J\min}}$ with $\delta_{z}>0$ a design constant, it can be further obtained that for $B_{2}\neq 0$, we have:
	\begin{equation}\label{dothz}
		\dot{h}(t) > -\gamma h(t) + \delta_{z}\|\boldsymbol{z}_{2}\|
	\end{equation}
	Considering different conditions of $h(t)$'s initial value $h(0)$:
	\begin{enumerate}
	\item For $h(0)>0$, according to the relationship given by equation (\ref{dothz}), this indicates that the time response of $h(t)$ has been confined within the tube region and will be maintained, as per Lemma \ref{T1}.
	
	\item For $h(0)\leq 0$, it will take some time for $h(t)$ to become positive again. During this period, the following condition holds:
	\begin{equation}
		h(t) = K_{h}\left(\Delta^{2}_{h}-\|\boldsymbol{z}_{2}\|^{2}\right) \le 0 \implies \|\boldsymbol{z}_{2}\| \ge \Delta_{h}
	\end{equation}
	Therefore, the relationship given by equation (\ref{dothz}) becomes:
	\begin{equation}\label{dothsmallzero}
		\dot{h}(t) > -\gamma h(t) + \delta_{z} \Delta_{h}
	\end{equation}
	Integrating both sides of equation (\ref{dothsmallzero}) and applying the Comparison Lemma (Gronwall Inequality) \cite{khalil2002control}, we obtain:
	\begin{equation}\label{hcondition}
		h(t) > \left[h(0) - \frac{\delta_{z} \Delta_{h}}{\gamma}\right] e^{-\gamma t} + \frac{\delta_{z} \Delta_{h}}{\gamma}
	\end{equation}
	Considering the zeroing point of equation (\ref{hcondition}), let $T_{h}$ satisfy:
	\begin{equation}\label{Th1}
		\left[h(0) - \frac{\delta_{z} \Delta_{h}}{\gamma}\right] e^{-\gamma T_{h}} + \frac{\delta_{z} \Delta_{h}}{\gamma} = 0,
	\end{equation}
	it can be derived as:
	\begin{equation}
		T_{h} = \frac{1}{\gamma} \ln\left(\frac{\gamma |h(0)|}{\delta_{z} \Delta_{h}} + 1\right)
	\end{equation}

	\end{enumerate}

Therefore, for $h(0) \leq 0$, $h(t)$ will eventually satisfy $h(t) > 0$, ensuring that $\|\boldsymbol{z}_{2}\| < \Delta_{h}$ will be achieved for $t \in [T_{h}, +\infty]$, where $T_{h}$ is a time instant given by equation (\ref{Th1}).

\subsubsection{Attractivity Anslysis of the Tube Region $H(t)>0$}

Building upon the attractivity analysis of $\boldsymbol{z}_{2}$ to its tube region defined by $h(t)>0$, we now pursue our focus to the convergence behavior of $H(t)$, where $H(t)$ is defined in equation (\ref{Hfunction}), signifying the satisfaction of the tracking performance constraint.

 Taking the time-derivative of $H(t)$, it can be derived that:
\begin{equation}
	\begin{aligned}
		\dot{H}(t) &= -2K_{H}\boldsymbol{s}^{\text{T}}\left[\boldsymbol{F}_{e}\boldsymbol{\omega}_{e}-\dot{\boldsymbol{\rho}}\right]
	\end{aligned}
\end{equation}
Note that $\boldsymbol{\omega}_{e} = \boldsymbol{\omega}_{v}+\boldsymbol{z}_{2}$, substituting the design of $\boldsymbol{\omega}_{v}$ given by equation (\ref{omegav}) into the expression, this gives rise to:
\begin{equation}
	\dot{H}(t) = -2K_{H}\boldsymbol{s}^{\text{T}}\left(2\lambda_{v}K_{H}-K_{s}\right)\boldsymbol{s}-2K_{H}\boldsymbol{s}^{\text{T}}\boldsymbol{F}_{e}\boldsymbol{z}_{2}
\end{equation}
Following a similar analysis, if $B_{1}=0$, the system has reached the equilibrium point $\boldsymbol{s}=0$. Therefore, we will focus on the general condition where $B_{1}\neq 0$:

Since $B_{1} = 4K^{2}_{H}\boldsymbol{s}^{\text{T}}\boldsymbol{s}$, owing to the fact that $\|\boldsymbol{F}_{e}\|=\frac{1}{2}$, we can derive:
\begin{equation}
	\dot{H}(t) > -B_{1}\lambda_{v}-K_{H}\|\boldsymbol{s}\|\|\boldsymbol{z}_{2}\|
\end{equation}
Meanwhile, note that $-B_{1}\lambda_{v} = A_{1}+\sqrt{A^{2}_{1}+q_{1}(B_{1})B_{1}}>A_{1}$ holds for $B_{1}\neq 0$, recalling the definition of $A_{1}$ as shown in equation (\ref{A1B1}), it can be obtained that:
\begin{equation}\label{dotHcondition}
	\begin{aligned}
		\dot{H}(t) 
		&> -\alpha H(t) + \delta_{H}\|\tanh(C_{s}\boldsymbol{s})\|-K_{H}\|\boldsymbol{z}_{2}\|\|\boldsymbol{s}\|\\
		&>-\alpha H(t) + \left(\delta_{H}-K_{H}\|\boldsymbol{z}_{2}\|\right)\|\boldsymbol{s}\|
	\end{aligned}
\end{equation}
note that the above relationship can be derived since $\|\tanh(C_{s}\boldsymbol{s})\|>\|\boldsymbol{s}\|$ holds for sufficiently large constant $C_{s}$, according to Lemma \ref{lemmatanh}.
Therefore, it can be observed that the dynamics of $H(t)$ depends on the value of $\|\boldsymbol{z}_{2}\|$, hence the following analysis will be given considering different conditions of $h(0)$ and $H(0)$:
\begin{enumerate}
	\item For $h(0)>0$ and $H(0)>0$, according to the aforementioned analysis, we have $\|\boldsymbol{z}_{2}\| < \Delta_{h}$. Choosing the constant $\delta_{H}$ as $\delta_{H} = \delta_{S} + K_{H}\Delta_{h}$ with $\delta_{S} > 0$ a positive constant, it can be easily obtained from equation (\ref{dotHcondition}) that $\dot{H}(t) > -\alpha H(t) + \delta_{S}\|\boldsymbol{s}\| > -\alpha H(t)$ holds for $H(0)>0$. As per the PAP condition, this shows that the state trajectory of $\boldsymbol{s}$ will continue to be restricted in the tube region for $t\in[0,+\infty)$, satisfying $H(t)>0$. Such a condition is illustrated in the Subfigure 1 of Figure \ref{PAPcondition1}.
	
	\begin{figure}[hbt!]
	\centering 
	\includegraphics[width = 0.5\textwidth]{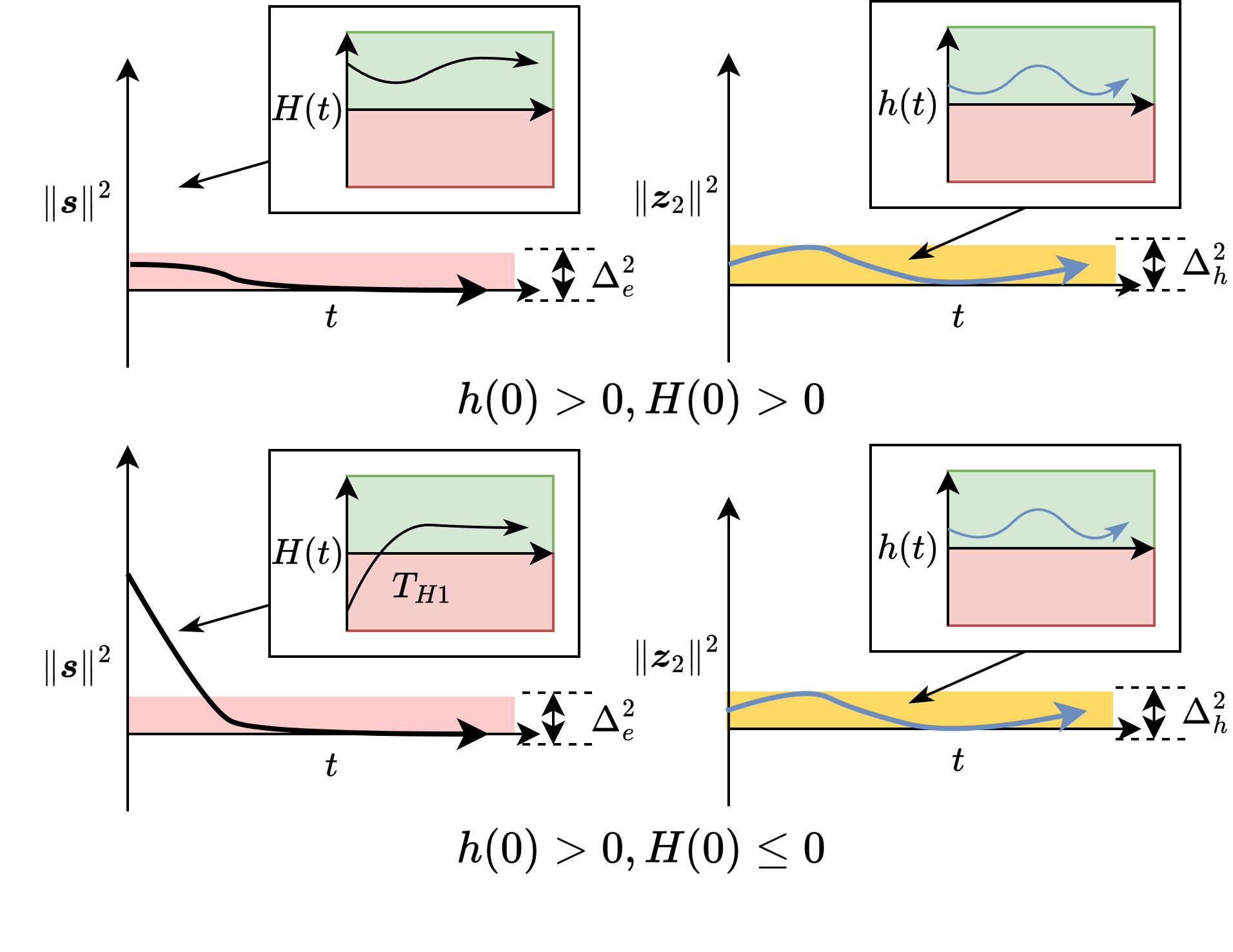}
	\caption{Illustration Explanation of the Condition 1 and 2}    
	\label{PAPcondition1}  
\end{figure}

\item For $h(0)>0$ but with $H(0)\leq 0$, this indicates that $\|\boldsymbol{s}(0)\|\ge\Delta_{e}$, meaning the system is performance-unsatisfied at the initial phase. However, by choosing $\delta_{H} = \delta_{S}+K_{H}\Delta_{h}$, where $\delta_{S}$ is a positive constant, it can be derived that:
\begin{equation}\label{dotH2}
	\dot{H}(t) > -\alpha H(t) + \delta_{S}\Delta_{e}
\end{equation}
Therefore, the time response of $H(t)$ satisfies:
\begin{equation}\label{Hcondition1}
	H(t) > \left[H(0)-\frac{\delta_{S}\Delta_{e}}{\alpha}\right]e^{-\alpha t} + \frac{\delta_{S}\Delta_{e}}{\alpha}
\end{equation}
Considering an time instant $T_{H1}$ that satisfies the equation: $\left[H(0)-\frac{\delta_{S}\Delta_{e}}{\alpha}\right]e^{-\alpha T_{H1}}+\frac{\delta_{S}\Delta_{e}}{\alpha} = 0$, then $T_{H1}$ can be expressed as:
\begin{equation}\label{TH}
	T_{H1} = \frac{1}{\alpha}\ln\left(\frac{\alpha|H(0)|}{\delta_{S}\Delta_{e}} + 1\right)
\end{equation}
Consequently, this indicates that for $t\in\left[T_{H1},+\infty\right)$, the system will satisfy $H(t)>0$, where $T_{H1}$ is specified by equation (\ref{TH}). Such a condition is illustrated in the Subfigure 2 of Figure \ref{PAPcondition1}.

\item For the case where $h(0) \leq 0$, we have proved that $h(t)$ will eventually satisfy $h(t) > 0$ after $t = T_{h}$. However, during the time interval $t \in [0, T_h]$, it should be noted that regardless of the initial condition $H(0)$ (either $H(0) > 0$ or $H(0) \leq 0$), there is still a possibility for $H(T_h)$ to be either $H(T_{h}) > 0$ or $H(T_{h}) \leq 0$, as shown in Figure \ref{PAPcondition3}.
 Therefore, we discuss the \textbf{worst-case} scenario , where $H(t) \le 0$ holds throughout $t \in [0, T_h]$.

	\begin{figure}[hbt!]
	\centering 
	\includegraphics[width = 0.5\textwidth]{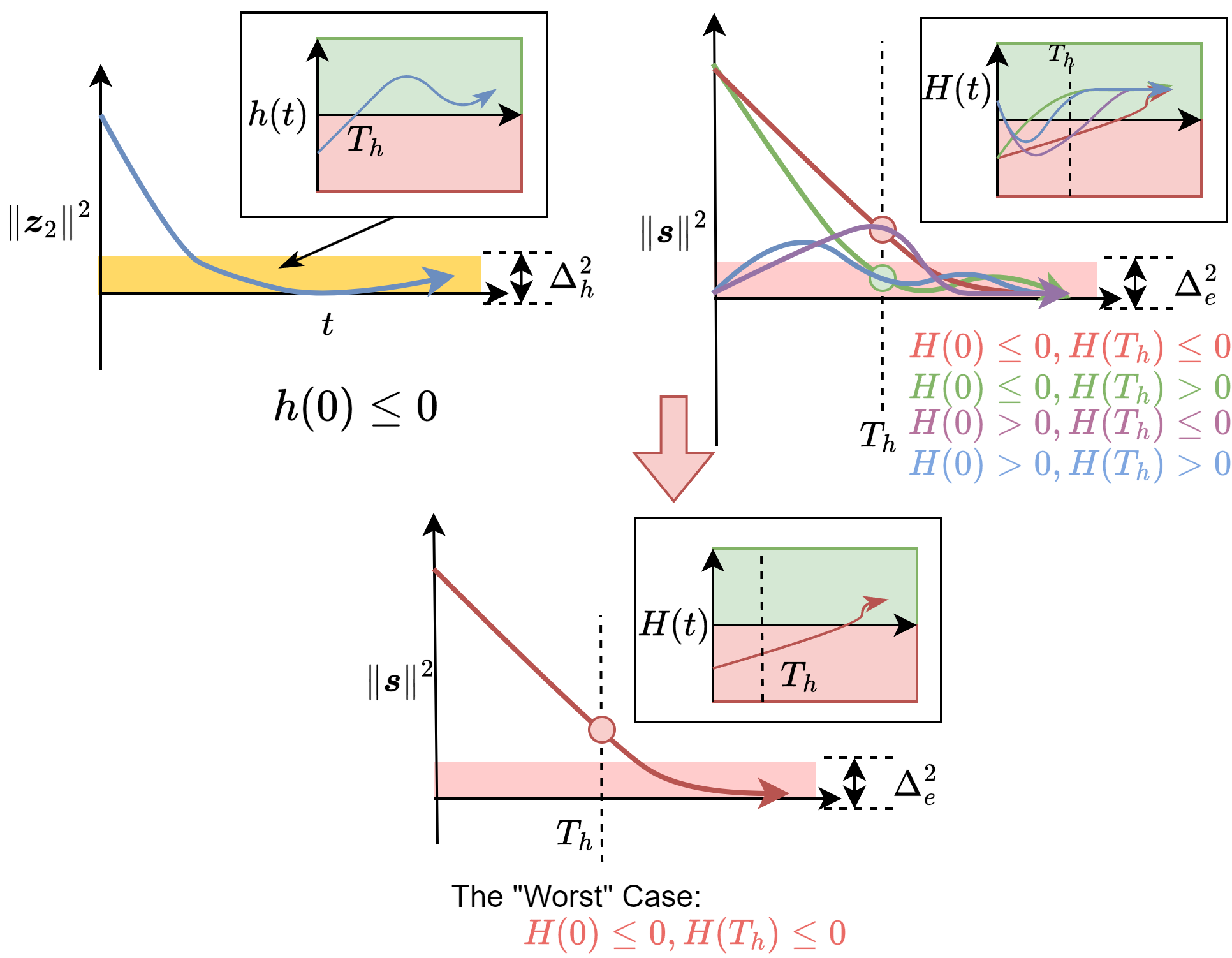}
	\caption{Illustration Explanation of the Condition 3 with Various Initial Conditions}    
	\label{PAPcondition3}  
\end{figure}

We start by discussing the behavior of $\|\boldsymbol{z}_{2}(t)\|$.
According to the aforementioned analysis, note that $\dot{h}(t) > -\gamma h(t) + \delta_{z}\Delta_{h}$ holds for $h(0)\le0$, hence we have:
\begin{equation}
	h(t) > \left[h(0) - \frac{\delta_{z}\Delta_{h}}{\gamma}\right]e^{-\gamma t} + \frac{\delta_{z}\Delta_{h}}{\gamma}
\end{equation}
Therefore, using the definition of $h(t)$, owing to the fact that $1 - e^{-\gamma t} \geq 0$, we have $\|\boldsymbol{z}_{2}(t)\|^{2} < \left(\|\boldsymbol{z}_{2}(0)\|^{2} - \Delta^{2}_{h}\right)e^{-\gamma t} + \Delta^{2}_{h}$, which gives:
\begin{equation}\label{z2response}
	\|\boldsymbol{z}_{2}(t)\| < Z_{1}e^{-\frac{\gamma t}{2}} + \Delta_{h}
\end{equation}
where $Z_{1}$ is defined as $Z_{1} \triangleq \sqrt{\|\boldsymbol{z}_{2}(0)\|^{2} - \Delta^{2}_{h}} > 0$ for brevity.

Meanwhile, according to equation (\ref{dotHcondition}), we have:
\begin{equation}
	\dot{H}(t) > -\alpha H(t)+(\delta_{S}+K_{H}\Delta_{h}-K_{H}\|\boldsymbol{z}_{2}\|)\|\boldsymbol{s}\|
\end{equation}

By combining this with the result in equation (\ref{z2response}), we find that $\delta_{S}+K_{H}(\Delta_{h}-\|\boldsymbol{z}_{2}\|)>\delta_{S}-K_{H}Z_{1}e^{-\frac{\gamma}{2}t}$.
Since we have assumed the worst case, where $H(t)\le0$ holds for $t\in[0,T_{h}]$, hence we always have $\|\boldsymbol{s}\|\ge\Delta_{e}$. Sorting these results, it leads to the following condition:
\begin{equation}\label{dotHcondition2}
	\dot{H}(t) > -\alpha H(t) + (\delta_{S}-K_{H}Z_{1}e^{-\frac{\gamma}{2}t})\Delta_{e}
\end{equation}
For writing convenience, let $m\triangleq\delta_{S}\Delta_{e}$ and $n\triangleq K_{H}Z_{1}\Delta_{e}$, equation (\ref{dotHcondition2}) can be rearranged as:
\begin{equation}\label{dotHnm}
	\dot{H}(t) > -\alpha H(t) + m - ne^{-\frac{\gamma}{2}t}
\end{equation}
The relationship shown in equation (\ref{dotHnm}) allows us to further discuss the lower bound of $H(t)$ on $t\in\left[0,T_{h}\right]$. The following discussion is given considering different relationships between $2\alpha$ and $\gamma$:

\begin{itemize}
	\item We first consider the general case, where $2\alpha \neq \gamma$, by applying the Comparison lemma and integrating both sides of equation (\ref{dotHnm}), it can be obtained that:
	\begin{equation}\label{Hlower}
		\begin{aligned}
			H(t) &> H(0)e^{-\alpha t} + \frac{m}{\alpha}(1-e^{-\alpha t}) \\
			&\quad - \frac{2n}{2\alpha-\gamma}(e^{-\frac{\gamma}{2}t}-e^{-\alpha t})
		\end{aligned}
	\end{equation}
From the above expression, it can be observed that the part " $\frac{m}{\alpha}(1-e^{-\alpha t})+H(0)e^{-\alpha t}$" is monotonically increasing. Meanwhile, note that $\frac{2n}{2\alpha-\gamma}(e^{-\frac{\gamma}{2}t}-e^{-\alpha t}) \geq 0$ permanently holds, it reaches its maxima at $t = \frac{1}{\frac{\gamma}{2}-\alpha}\ln\left(\frac{\gamma}{2\alpha}\right)$, while its maxima is given as:
	$\frac{n}{\alpha-\frac{\gamma}{2}}\exp\left(-\frac{\gamma}{\gamma-2\alpha}\ln\frac{\gamma}{2\alpha}\right)\left(1-\frac{\gamma}{2\alpha}\right)$.

	\item For the special case where $\gamma = 2\alpha$, the expression of equation (\ref{dotHnm}) becomes:
	\begin{equation}\label{dotH2gamma}
		\dot{H}(t) >-\alpha H(t)+m-ne^{-\alpha t}
	\end{equation}
	Similarly, by integrating on both sides of equation (\ref{dotH2gamma}), it can be yielded that:
	\begin{equation}
		\begin{aligned}
			H(t) > H(0)e^{-\alpha t}+\frac{m}{\alpha}\left(1-e^{-\alpha t}\right)-nte^{-\alpha t}
		\end{aligned}
	\end{equation}
For the tail term expressed as $n\cdot te^{-\alpha t}$, it reaches the maxima at $t=\frac{1}{\alpha}$, where $\max(nte^{-\alpha t})=\frac{1}{\alpha e}$ holds.
\end{itemize}

Consequently, sorting the results of these two conditions, we define $\mathcal{G}_{B}$ as:
\begin{equation}\label{GB}
	\begin{aligned}
		\mathcal{G}_{B} \triangleq 
		\begin{cases}
			\frac{n}{\alpha-\frac{\gamma}{2}}\exp\left(-\frac{\gamma}{\gamma-2\alpha}\ln\frac{\gamma}{2\alpha}\right)\left(1-\frac{\gamma}{2\alpha}\right),\quad\gamma\neq2\alpha\\
			\frac{1}{\alpha e},\quad\gamma=2\alpha
		\end{cases}
	\end{aligned}
\end{equation}
then it can be pointed out that:
\begin{equation}\label{Hlowerall}
	H(T_{h}) > H(0)e^{-\alpha T_{h}}+\frac{m}{\alpha}\left(1-e^{-\alpha T_{h}}\right)-\mathcal{G}_{B}\triangleq \mathcal{H}_{B}
\end{equation}
To facilitate the following analysis, we define the right side of equation (\ref{Hlowerall}) as $\mathcal{H}_{B}$, signifying the lower bound function of $H(t)$ at $t=T_{h}$.

	\begin{figure}[hbt!]
	\centering 
	\includegraphics[width = 0.5\textwidth]{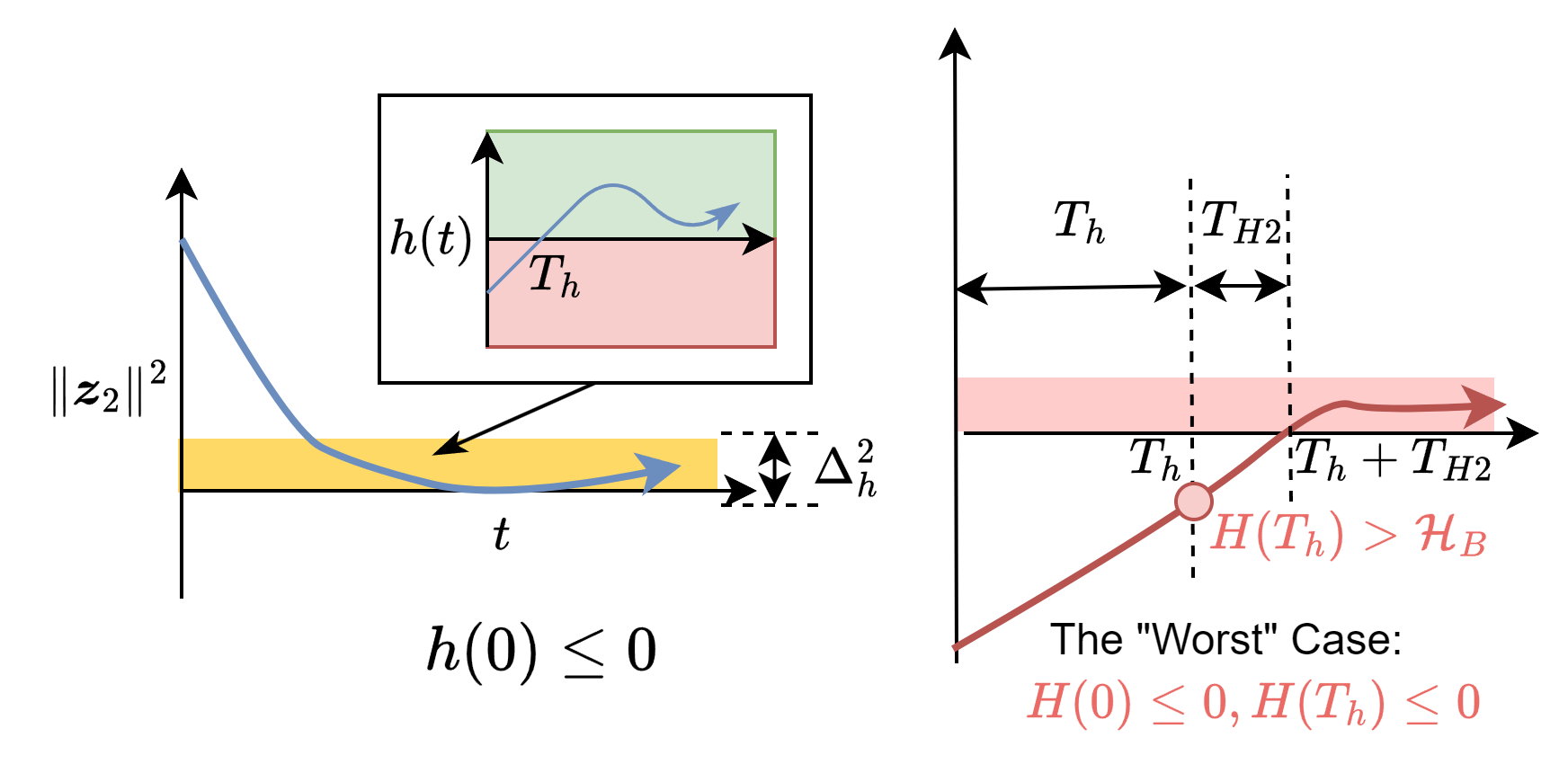}
	\caption{Illustration Explanation of $T_{h}$ and $T_{H2}$
		%		: $T_{h}$ stands for the upper bound of the time duration for $h(t)$ to converge to the tube; $T_{H2}$ stands for the upper bound of the time duration for $H(t)$ to converge from its possible lowest position at $t=T_{h}$, to the tube region
	}    
	\label{PAPconditionlast}  
\end{figure}
According to our analysis in Subsection \ref{attracth}, $h(t)>0$ will be satisfied after $t=T_{h}$, even under the worst circumstances. This implies that $\|\boldsymbol{z}_{2}\|$ has converged to be sufficiently small, satisfying $\|\boldsymbol{z}_{2}\| < \Delta_{h}$. Thus, under such circumstances, by applying a similar analysis as shown in the second condition (referring to equation (\ref{dotH2})), it can be observed that $H(t)$ will increase from its position $H(T_{h})\le0$ since then, requiring at most an additional time $T_{H2}$ to satisfy $H(t)>0$, where $T_{H2}$ can be given as:
\begin{equation}
	\begin{aligned}
		T_{H2} &= \frac{1}{\alpha}\ln\left(\frac{\alpha|\mathcal{H}_{B}|}{\delta_{S}\Delta_{e}} + 1\right)
	\end{aligned}
\end{equation}
Consequently, for $h(0) \leq 0$, even under the worst-case scenario, the system will enter the tube region, satisfying $H(t) > 0$ for $t \in [T_{h}+T_{H2},+\infty)$.

In fact, "$T_{h}$" represents the upper bound of the time duration for $h(t)$ to converge from its initial position $h(0)\le0$ to $h(t)>0$, while "$T_{H2}$" represents the upper bound of the time duration for $H(t)$ to converge from its position at $T =T_{h}$, i.e., $H(T_{h})\le 0$, to satisfy $H(t)>0$. Additional explanation is presented in Figure \ref{PAPconditionlast}.
\end{enumerate}

This completes the proof of Theorem \ref{T2}.

\subsection{Discussion of Parameter Selection}
In this subsection, we discuss about the parameter selection suggestions of the proposed PAP control scheme. It should be noted that the parameters related to $A_{1}$, $A_{2}$, $B_{1}$ and $B_{2}$ are in the most key roles of the proposed control scheme.

Firstly, from equation (\ref{dothz}), it can be observed that $\gamma>0$ controls the varying rate of $h(t)$, as we have:
	\begin{equation}
	\dot{h}(t) > -\gamma h(t) + \delta_{z}\|\boldsymbol{z}_{2}\|
\end{equation}
Define the right side of the inequality as $h_{r} \triangleq -\gamma h(t) + \delta_{z} \|\boldsymbol{z}_{2}\|$. It can be observed that for $h(t) > 0$, a larger $\gamma$ results in faster convergence of $h_{r}(t)$, hence allowing $h(t)$ to approach the boundary defined by $h(t) = 0$ less conservatively. When $h(t) \leq 0$, a larger $\gamma$ results in a faster increase of $h_{r}(t)$, and hence $h(t)$ will also increase faster.

For practical considerations, increasing the values of $\gamma$, $\sigma_{2}$, and $K_{h}$ leads to a larger $A_{2}$ and $q_{2}(B_{2})$, which in turn increases the value of the varying gain: $|\lambda_{u}|$. Consequently, this significantly enhances the controller gain term $2\lambda_{u}K_{h} - K_{2}$ that governs the $\boldsymbol{z}_{2}$-system, which enables a more aggressive approach for the state trajectories $z_{2i}(t)$ toward their corresponding tube regions defined by $h(t) > 0$. Also, these parameters should yet not be set too small, as doing so may lead to performance degradation, particularly when facing perturbations. This suggests that a balance must be considered when selecting these parameters. A similar analysis can be applied to $\alpha$, $\sigma_{1}$, and $K_{H}$, which affect the function $H(t)$ in a similar manner.

Note that compared to the original expression of $A(t)$ as shown in \cite{sontag1989universal}, we have added a term expressed as $\delta_{H}\|\tanh(C_{s}\boldsymbol{s})\|$ in the design of $A_{1}(t)$. Since $\|\boldsymbol{s}\| \geq \Delta_{e}$ holds for $H(t) \leq 0$, this introduces a constant convergence term as $\delta_{H}\|\tanh(C_{s}\boldsymbol{s})\| \geq \delta_{H}\|\boldsymbol{s}\|\ge \delta_{H}\Delta_{e}$ holds. This modification allows us to discuss the upper bound of the time duration for $H(t)$ to converge and satisfy $H(t) > 0$. Notably, the parameter $\delta_{H}$ should be balanced, as an overly large $\delta_{H}$ may lead to a chattering problem due to overly large $|\lambda_{v}|$. Similar discussion is also applicable to $\delta_{h}$.

% From the above analysis, 

\section{Simulation and Analysis}
In this section, we present several groups of simulation results to validate the effectiveness of the proposed PAP control scheme. 
%Firstly, we demonstrate the fundamental effectiveness of the proposed method in an attitude tracking control scenario with normal magnitudes of perturbations. Next, we focus on the system's behavior under strong sudden disturbances, presenting a robustness validation considering non-zero initial angular velocity and severe sudden perturbations. A comparison section is further presented to showcase the superiority of the proposed method in singularity-free operation and its capability to achieve precisely-assigned performance criteria, compared with typical PPC control schemes. Finally, we conduct a Monte Carlo simulation campaign to demonstrate the utility of the proposed method under various conditions.

%\subsection{Control Scenario Establishment}
We first discuss the establishment of control scenarios.
In this section, the spacecraft is assumed to be a rigid-body spacecraft, where the inertia parameter is given as:
\begin{equation}
	\begin{aligned}
		\boldsymbol{J} = \begin{bmatrix}
			2.8&0.1&0.5\\
			0.1&2.5&0.24\\
			0.5&0.24&1.9
		\end{bmatrix}\left(\text{kg}\cdot\text{m}^{2}\right)
	\end{aligned}
\end{equation}
The normal external disturbance $\boldsymbol{d}\in\mathbb{R}^{3}$ is considered to be a time-varying periodic perturbation, expressed as follows:
\begin{equation}
	\boldsymbol{d} = 
	\begin{bmatrix}
		1\sin\left(3\omega_{\text{p}}t\right) + 4\cos\left(3\omega_{\text{p}}t\right) -20\\ 	
		5\sin\left(2\omega_{\text{p}}t\right) + 3\cos\left(3\omega_{\text{p}}t\right) +20\\ 	
		3\sin\left(2\omega_{\text{p}}t\right) - 1\cos\left(4\omega_{\text{p}}t\right) +20\\ 	
	\end{bmatrix}\times 10^{-4}\left(\text{N}\cdot\text{m}\right)
\end{equation}
where the angular rate parameter $\omega_{p}$ is set to $\omega_{p} = 0.01\text{rad}/\text{s}$. It is important to note that the magnitude of the disturbance on each axis is up to $2.5\times 10^{-3}\text{N}\cdot\text{m}$, which is much larger than an actual perturbation, and hence is enough for validating the robustness of the proposed method.

Although the control scheme does not explicitly consider the input saturation issue, we limit the maximum output control torque to $0.05\text{N}\cdot\text{m}$ for each axis, mainly for practical consideration.

Main parameters of the proposed PAP control scheme are given as follows: $K_{H} = 2$, $K_{h} = 1$, $K_{s} = 0.1$, $K_{2} = 2$, $\delta_{H} = 1\times 10^{-5}$, $C_{s} = 1\times 10^{7}$, $\delta_{h} = 2\times 10^{-3}$, $\alpha = 0.5$, $\beta = 1$, $\sigma_{1} = 0.05$, $\sigma_{2} = 1$.

\textbf{Precisely-Assigned Performance (PAP) Requirements: } The steady-state error should satisfy $|q_{evi}(t)| < 1\times 10^{-5}$, the desired settling time is $50\text{s}$, and the overshoot should be eliminated as small as possible. 

Our simulation is conducted using the MATLAB with a step of $0.1\text{s}$, and the control law calculation is performed outside the ODE45 integrator.

\subsection{Normal Case: Attitude Tracking Control}\label{normalcase}
Firstly, we demonstrate the fundamental effectiveness of the proposed method in an attitude tracking control scenario with normal magnitudes of perturbations. According to the given performance criteria, we set the width of the tube region as $\Delta_{e} = 1\times 10^{-5}$, $\Delta_{h} = 1\times 10^{-5}$, and the RPF is specified by $\rho_{i}(0) = q_{evi}(0) - 0.1$, $T_{sd} = 50\text{s}$, and $\rho_{\infty} = 0$.  
\begin{remark}
In order to validate that the state trajectory can be guided to the PAP-satisfied region with arbitrary initial conditions, we deliberately set the initial condition of the RPF as $\rho_{i}(0) = q_{evi}(0) - 0.1$. Such a design is to validate that the conservative assumption, which requires the state trajectory to be restricted in the performance-satisfied region as referenced in \cite{bechlioulis_robust_2008,bechlioulis_adaptive_2009}, is unnecessary for our PAP control scheme.
\end{remark}

The initial attitude quaternion is given as $\boldsymbol{q}_{s}(0) = \left[0.3482, 0.5222, 0.6963, 0.3482\right]^{\text{T}}$, and the initial angular velocity $\boldsymbol{\omega}_{s}$ is considered to be zero as $\boldsymbol{\omega}_{s}(0) = \boldsymbol{0}$. The desired reference attitude quaternion is determined by the initial target attitude quaternion $\boldsymbol{q}_{d}(0)$ and the reference angular velocity $\boldsymbol{\omega}_{d}$, given as:
\begin{equation}
	\begin{aligned}
		\boldsymbol{q}_{d}(0) &= \left[0, 0, 0, 1\right]^{\text{T}}\\
		\boldsymbol{\omega}_{d} &= 0.3 \times \left[\cos\left(\frac{t}{80}\right), \sin\left(\frac{t}{100}\right), -\cos\left(\frac{t}{100}\right)\right]^{\text{T}}
	\end{aligned}
\end{equation}
where the unit of $\boldsymbol{\omega}_{d}$ is $^{\circ}/\textbf{s}$.

\begin{figure}[hbt!]
	\centering 
	\includegraphics[width = 0.5\textwidth]{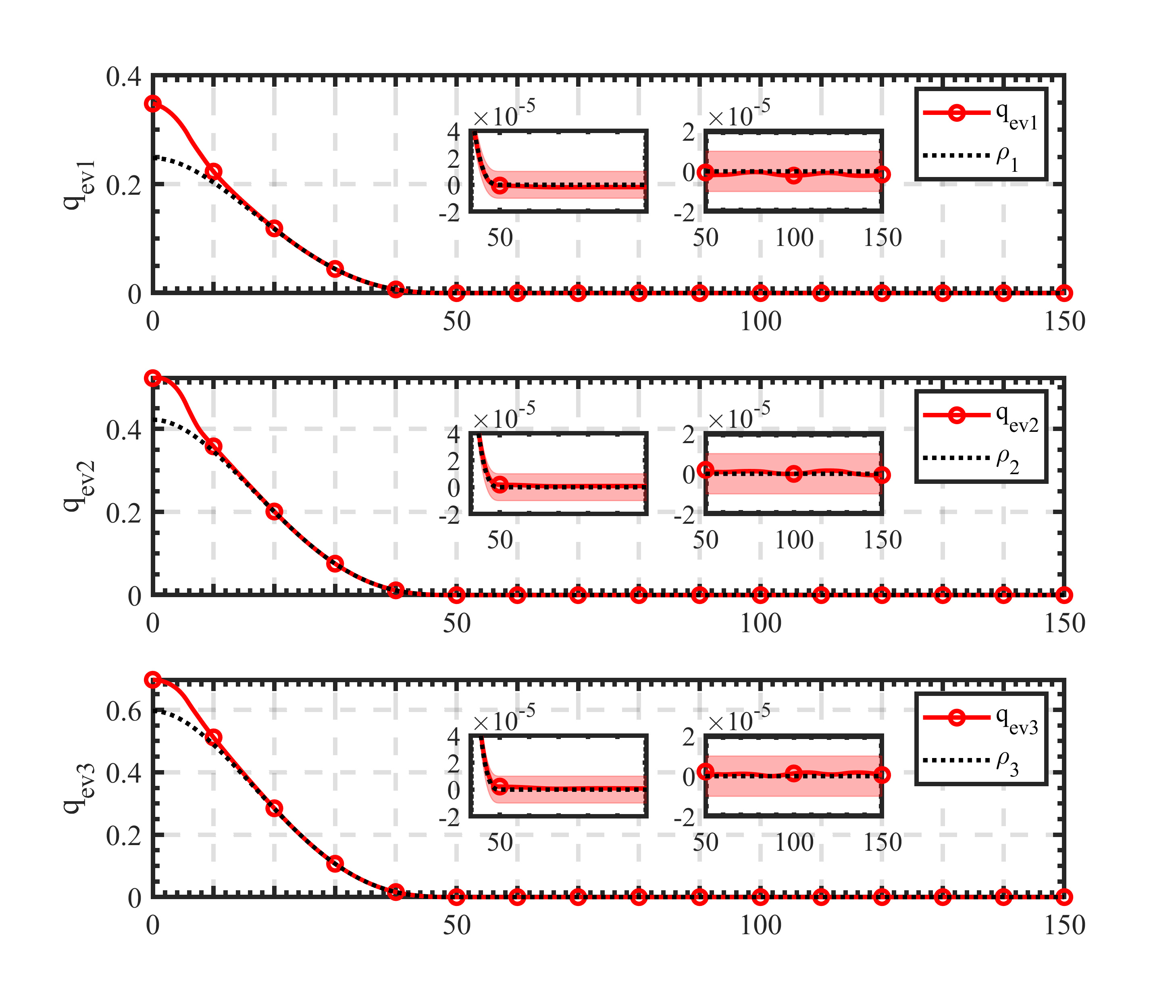}
	\caption{Time evolution of $q_{evi}(t)$ and $\rho_{i}(t)$ (Normal Case)}    
	\label{PAPC_QE}  
\end{figure}

The simulation results are presented in Figures \ref{PAPC_QE}, \ref{PAPC_WS}, \ref{PAPC_U}, and \ref{PAPC_HS}, respectively. 
Figure \ref{PAPC_QE} illustrates the time evolution of the attitude quaternion error $\boldsymbol{q}_{ev}$ and the RPF $\boldsymbol{\rho}$, denoted as $q_{evi}(t)$ and $\rho_{i}(t)$, respectively.
Each component of the time evolution of the attitude angular velocity $\omega_{si}(t)$ is then illustrated in Figure \ref{PAPC_WS}, while Figure \ref{PAPC_U} shows the time response of the control input $\boldsymbol{u}$. Figure \ref{PAPC_HS} further illustrates the time evolution of the Zeroing barrier functions, $H(t)$ and $h(t)$.

In Figure \ref{PAPC_QE}, the red-solid line represents the error response $q_{evi}(t)$, while the black-dotted line represents the RPF $\rho_{i}(t)$. The red-filled tube-like region represents the performance-satisfied region, specified by $|q_{evi}(t) - \rho_{i}(t)| < \Delta_{e}$. In Figure \ref{PAPC_WS}, the black-dotted line represents the time-evolution of the desired angular velocity $\boldsymbol{\omega}_{d}$.

From Figure \ref{PAPC_QE}, it can be observed that each $q_{evi}(t)$ rapidly approaches the RPF $\rho_{i}(t)$, implying the convergence of the tracking error variable $\boldsymbol{s}$.
In the zoomed-in figure of Figure \ref{PAPC_QE}, it can be observed that $\boldsymbol{q}_{ev}$ has converged and been confined within the tube region at around $t = 20$ s, ensuring that $\boldsymbol{q}_{ev}$ converges to the steady-state region at precisely $t = T_{sd} = 50$ s. Meanwhile, each $q_{evi}(t)$ is maintained to satisfy $|q_{evi}(t)|<1\times 10^{-5}$ for $t\in\left[T_{sd},+\infty\right)$. Therefore, this indicates that all the desired PAP constraints can be accurately satisfied.

From the results shown in Figure \ref{PAPC_WS}, the actual response of the attitude angular velocity $\omega_{si}(t)$ is able to track the given desired angular velocity $\omega_{di}(t)$. 
\begin{figure}[hbt!]
		\centering 
	\includegraphics[width = 0.5\textwidth]{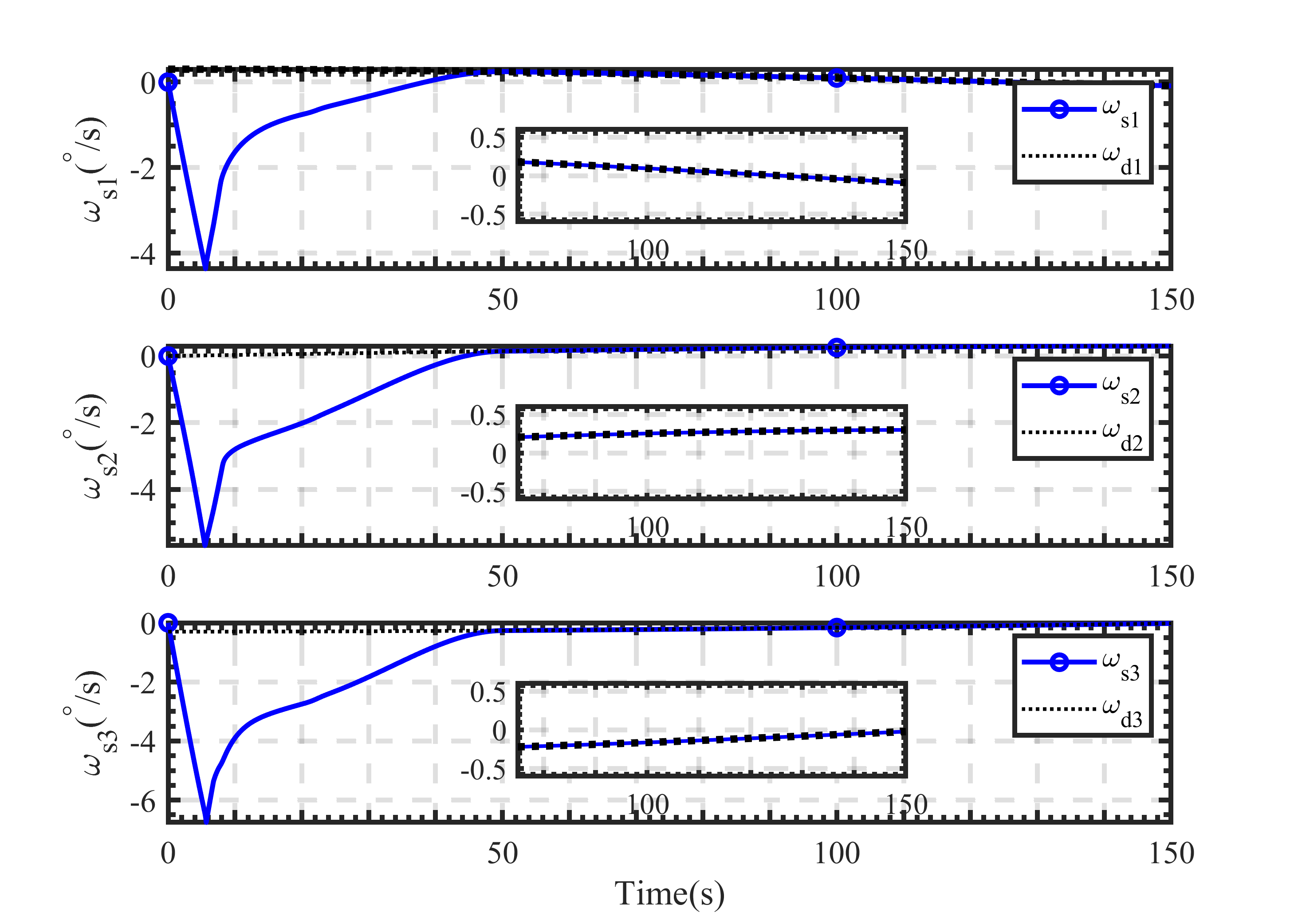}
	\caption{Time evolution of $\omega_{si}(t)$ (Normal Case)}    
	\label{PAPC_WS} 
	\centering 
	\includegraphics[width = 0.5\textwidth]{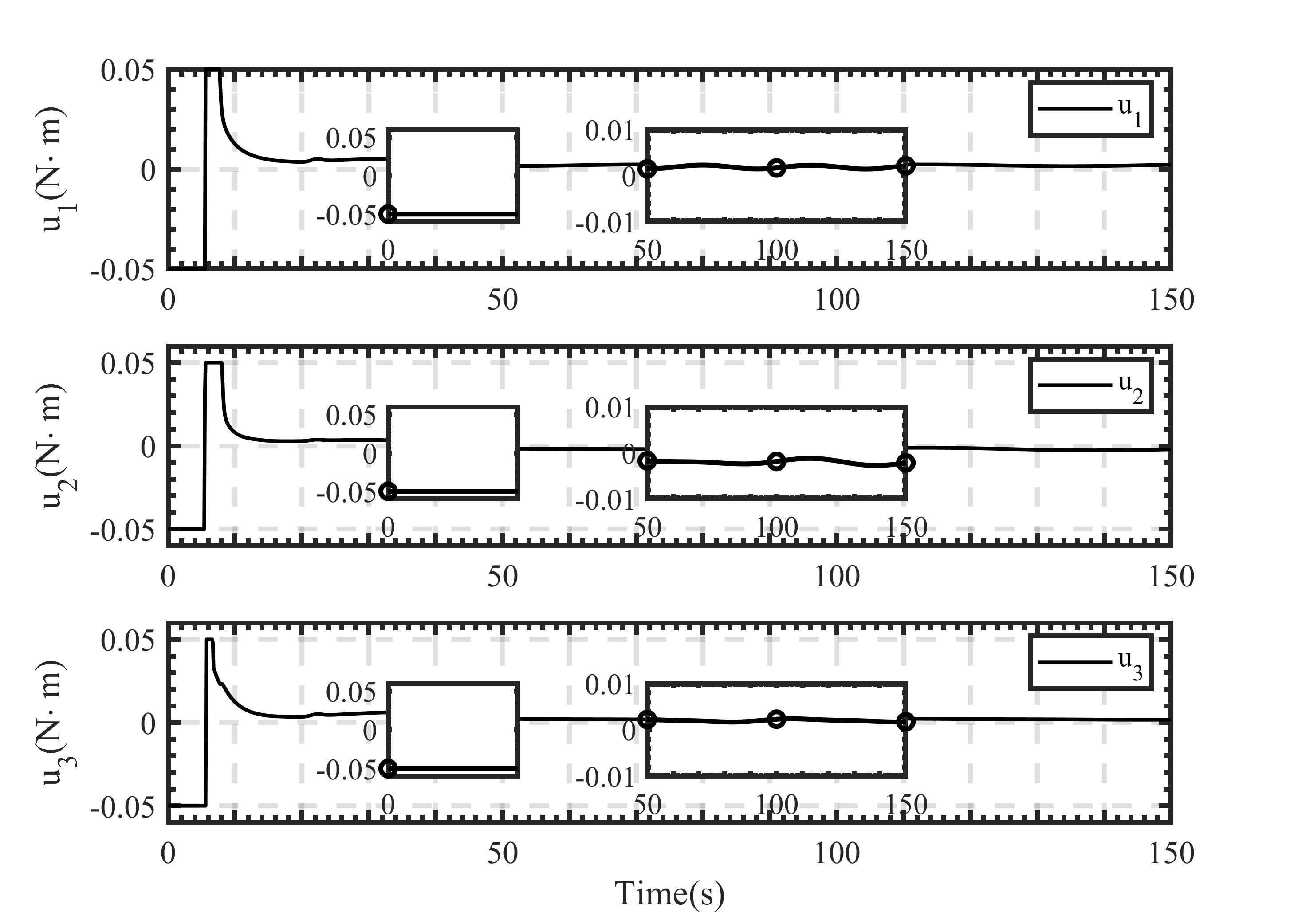}
	\caption{Time evolution of $u_{i}(t)$ (Normal Case)}    
	\label{PAPC_U} 
\end{figure}

From the results shown in Figure \ref{PAPC_U}, it can be observed that the calculated control input is relatively large at the initial phase, in order to rapidly eliminate the tracking error $\boldsymbol{z}_{2}$. However, once the error state tracks the reference performance function, the required control input decreases to be small, leads to mild control behavior.

\begin{figure}[hbt!]
				\centering 
	\includegraphics[width = 0.5\textwidth]{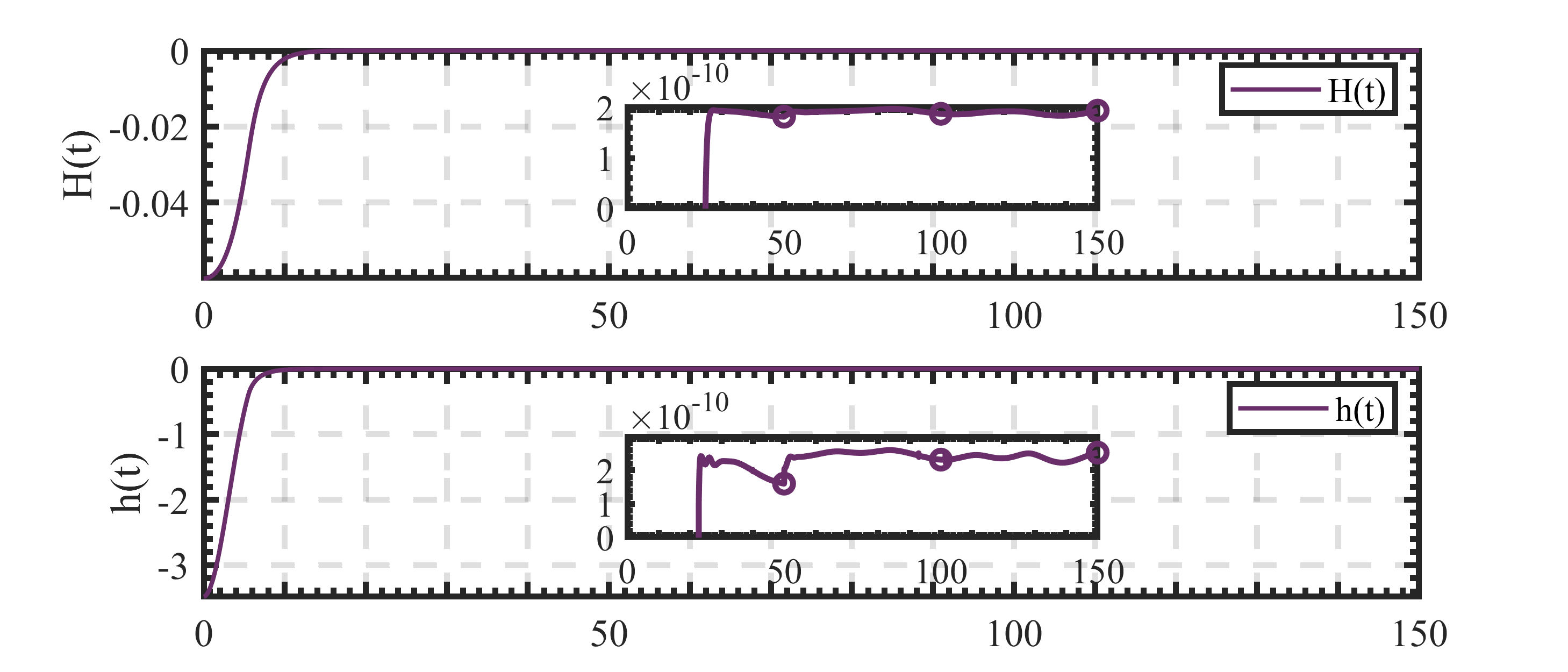}
	\caption{Time evolution of $H(t)$ and $h(t)$ (Normal Case)}    
	\label{PAPC_HS} 
\end{figure}
Further, from the time-evolution of $H(t)$ and $h(t)$ as shown in Figure \ref{PAPC_HS}, it can be observed that the time evolution of $h(t)$ rapidly converges, becoming positive at around $t = 20$ s. Once $h(t)$ becomes positive, it remains greater than 0 thereafter, as shown in the second zoom-in subfigure in Figure \ref{PAPC_HS}.

Regarding the time evolution of $H(t)$, it increases at a slower rate than $h(t)$. However, from the results shown in the first (upper) subfigure and zoom-in subfigure of Figure \ref{PAPC_HS}, $H(t) > 0$ holds permanently after it becomes positive.

Consequently, these results validate our analysis as provided in Subsection \ref{Attract}, showing that the designed controller is able to maintain or achieve the satisfaction of the tracking performance constraint as given by equation (\ref{trackingpercons}).

%\begin{figure}[hbt!]
%
%\end{figure}

\subsection{Robustness Validation}
Next, we focus on the system's behavior under strong sudden disturbances, presenting a robustness validation considering non-zero initial angular velocity and severe sudden perturbations, showing that the proposed method is naturally singularity-free. 
\begin{figure}[hbt!]
	\centering 
	\includegraphics[width = 0.5\textwidth]{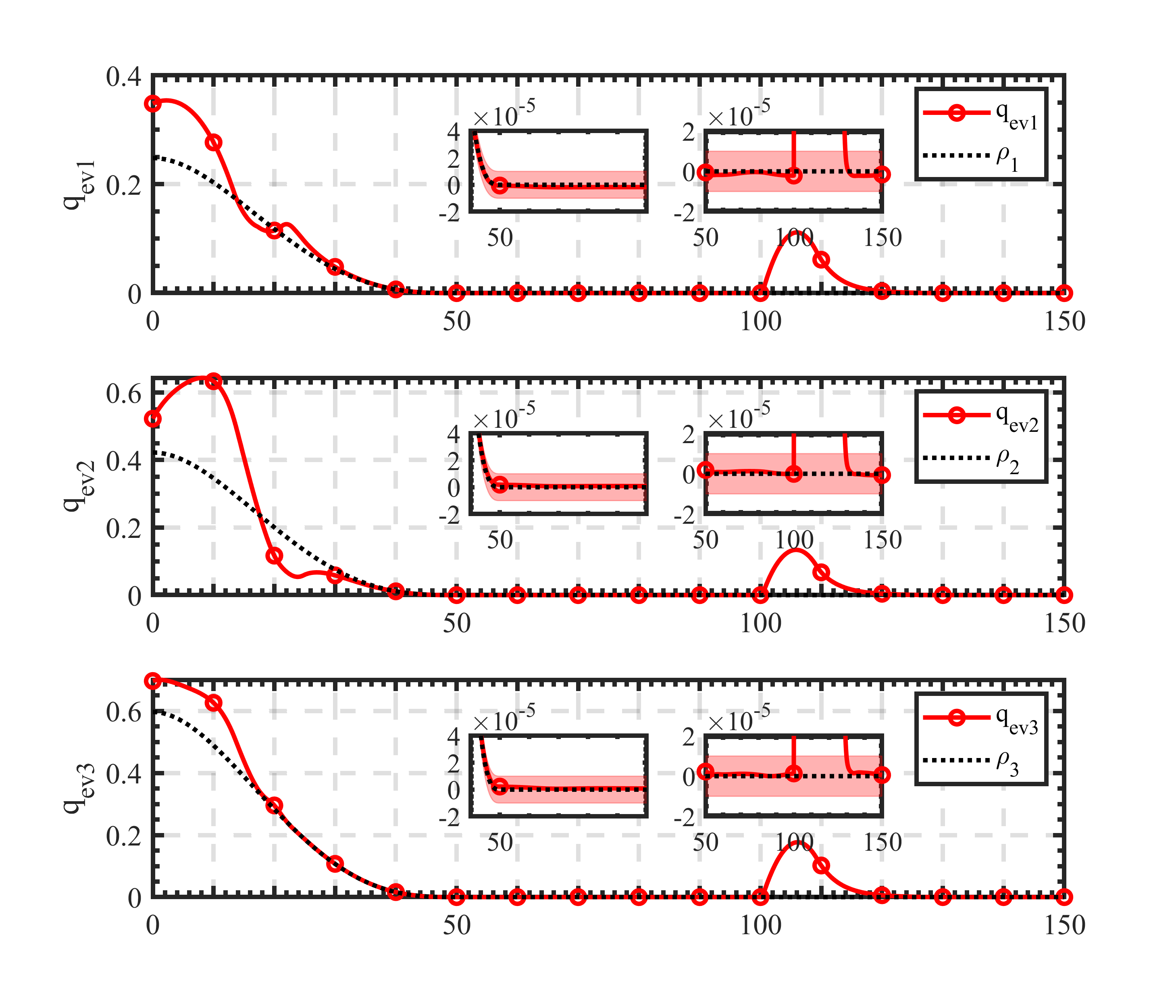}
	\caption{Time evolution of $q_{evi}(t)$ and $\rho_{i}(t)$ (Robust Validation)}    
	\label{PAPC_QEd}  
\end{figure}
The initial attitude is the same as in Subsection \ref{normalcase}, while the initial condition of the angular velocity is given as $\boldsymbol{\omega}_{s}(0) = 5\times\left[1,1,1\right]^{\text{T}}$, where the unit is $^\circ/\text{s}$. A large sudden disturbance $\boldsymbol{d}_{a} = 0.5\times \left[1,1,1\right]^{\text{T}}$ (N$\cdot$m) will be exerted on the system at $t = 100$ s, lasting for 0.5 s, hence heavily perturbing the spacecraft's attitude.

We focus on the time evolution of the attitude quaternion error $\boldsymbol{q}_{ev}$, the RPF $\boldsymbol{\rho}$, the angular velocity $\boldsymbol{\omega}_{s}$, the control input $\boldsymbol{u}$, and the time evolution of $h(t)$ and $H(t)$, illustrated in Figures \ref{PAPC_QEd}, \ref{PAPC_WSd}, \ref{PAPC_Ud}, and \ref{PAPC_HSd}, respectively.

From Figure \ref{PAPC_QEd}, it can be observed that since the initial attitude angular velocity is nonzero, it takes around $42\text{s}$ for the system to guide error states $q_{evi}(t)$ converge to the RPF $\rho_{i}(t)$. 
After the sudden disturbance is exerted at $t = 100\text{s}$, it can be observed that the attitude is heavily perturbed. However, shortly after $t = 125\text{s}$, the PAP control scheme re-stabilizes the system and still achieves the desired control accuracy, satisfying $|q_{evi}(t)| < 1\times 10^{-5}$, as shown in the zoom-in subfigures of Figure \ref{PAPC_QEd}.

\begin{figure}[hbt!]
		\centering 
	\includegraphics[width=0.5\textwidth]{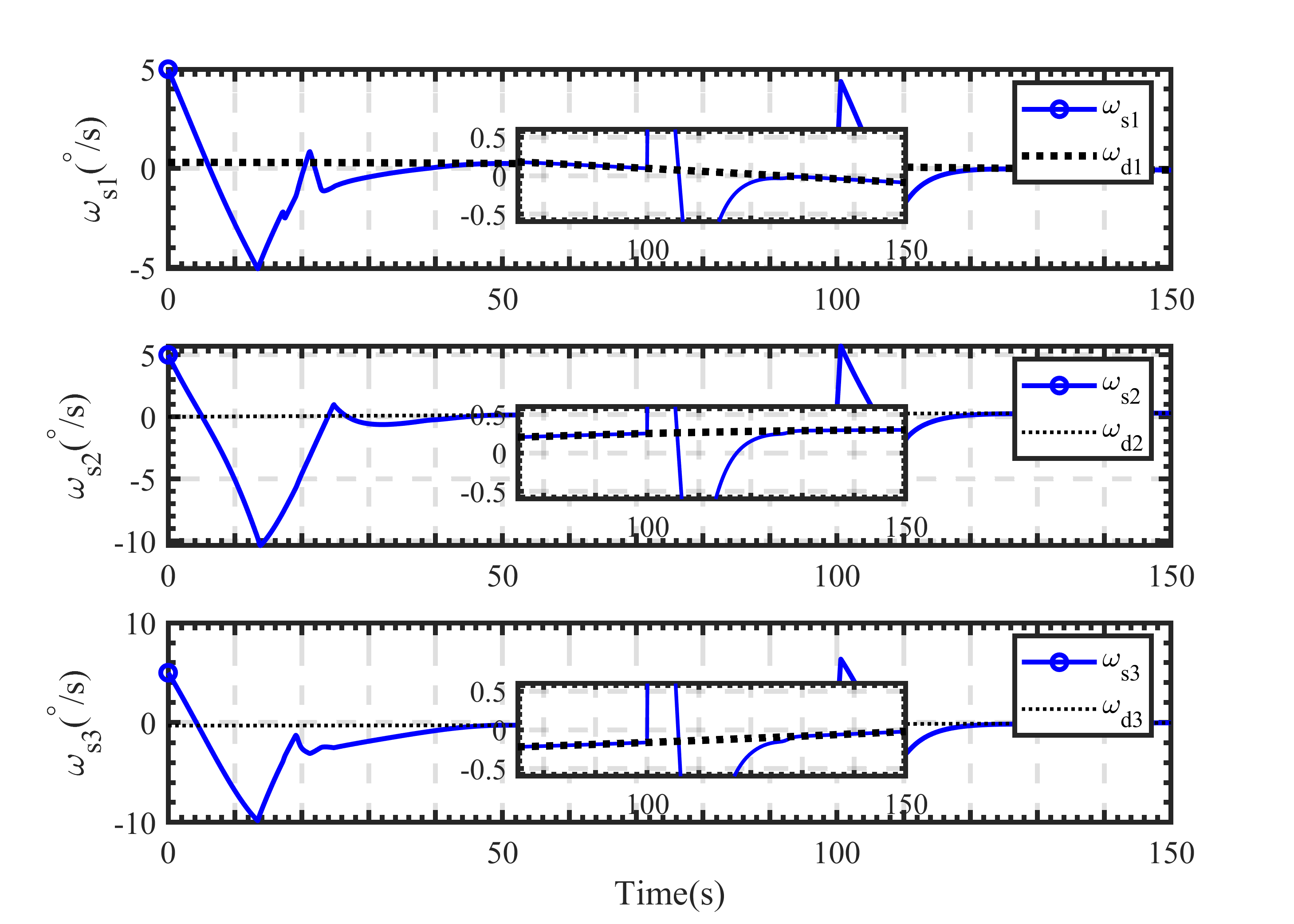}
	\caption{Time evolution of $\omega_{si}(t)$ (Robust Validation)}    
	\label{PAPC_WSd} 
			\centering 
\includegraphics[width=0.5\textwidth]{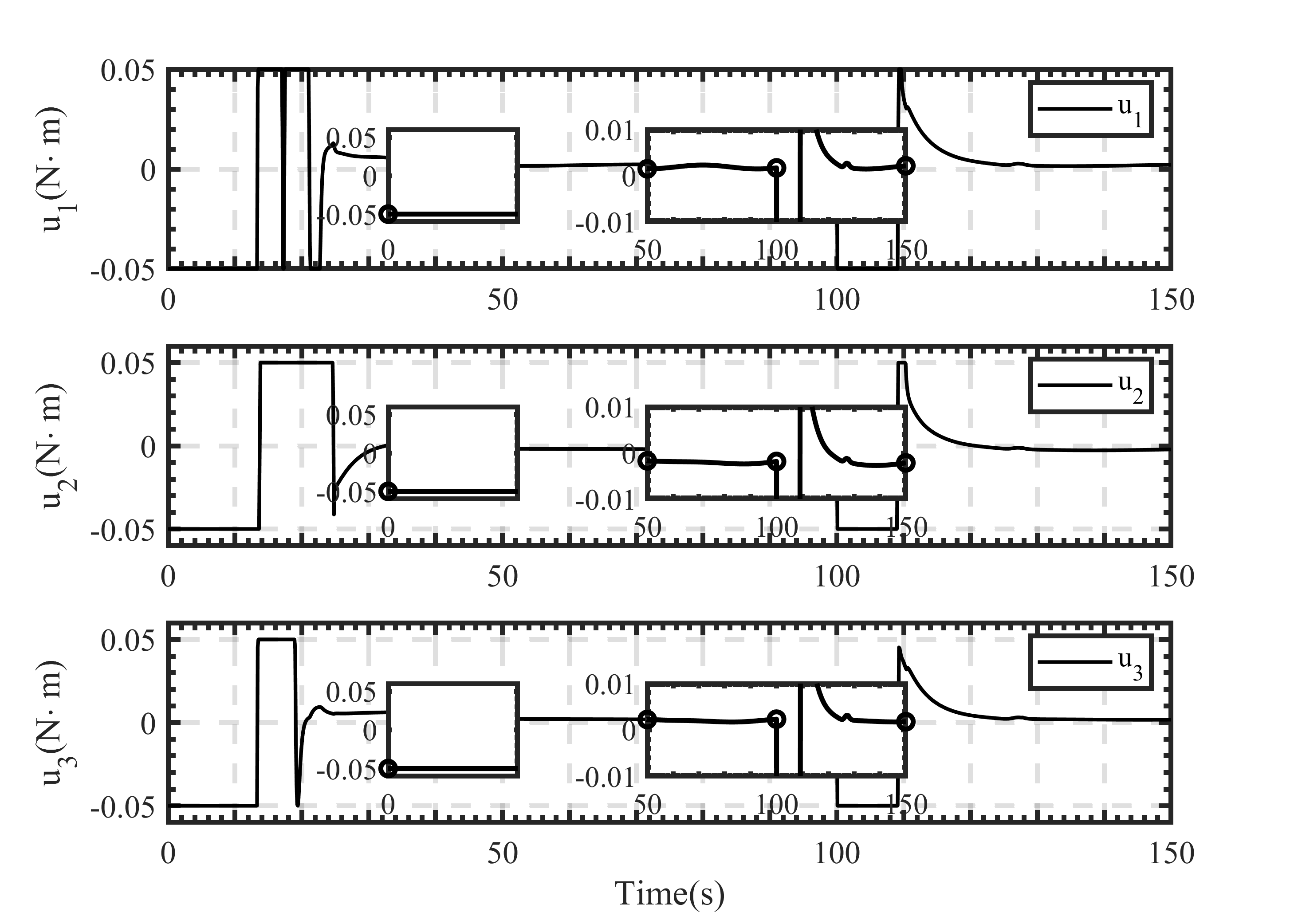}
\caption{Time evolution of $u_{i}(t)$ (Robust Validation)}    
\label{PAPC_Ud} 
\end{figure}

From the results of $\omega_{si}(t)$ and $u_{i}(t)$, as shown in Figures \ref{PAPC_WSd} and \ref{PAPC_Ud}, the nonzero initial attitude angular velocity condition requires the system to exert its maximum output torque of $0.05\text{N}\cdot\text{m}$ to stabilize the spacecraft. Also, it can be observed that the proposed controller stabilizes the system shortly after the occurrence of severe perturbation.
\begin{figure}[hbt!]
			\centering 
	\includegraphics[width=0.5\textwidth]{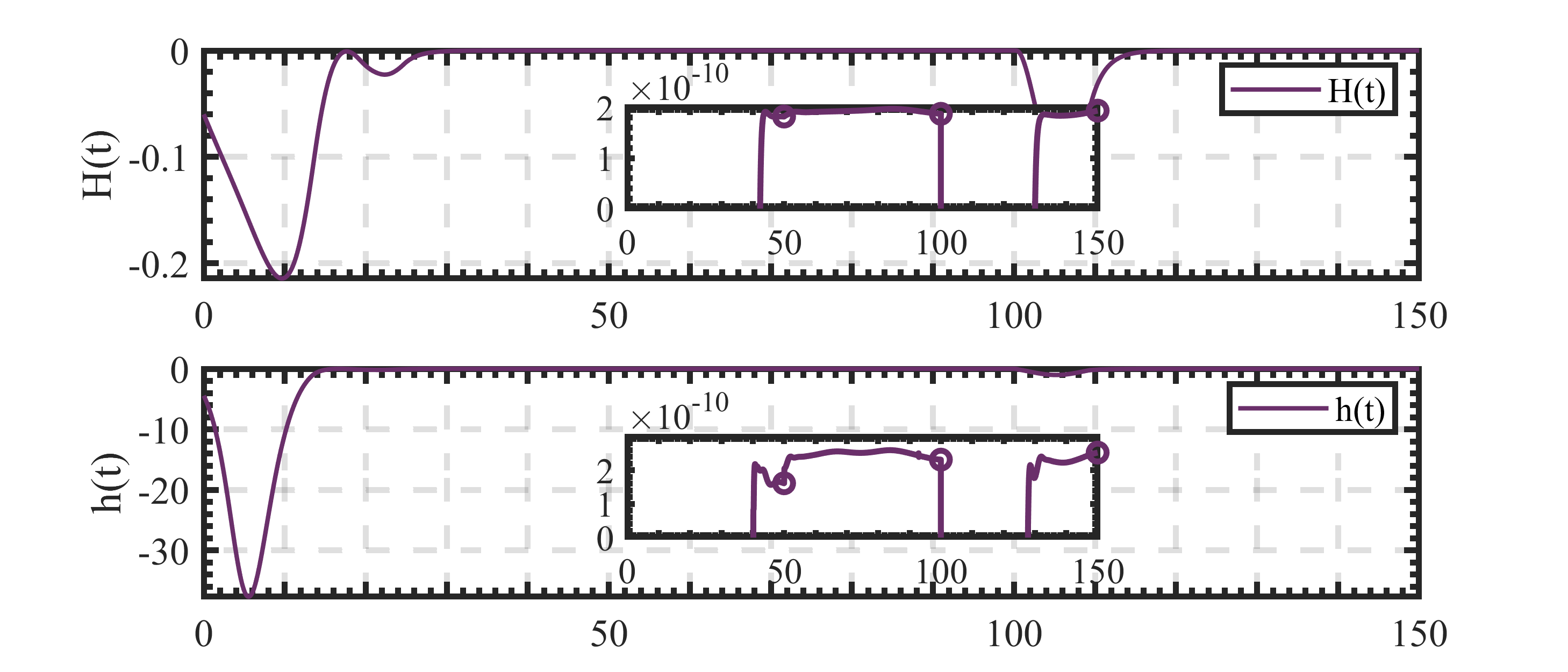}
	\caption{Time evolution of $H(t)$ and $h(t)$ (Robust Validation)}    
	\label{PAPC_HSd} 
\end{figure}

As shown in Figure \ref{PAPC_HSd}, even under a non-ideal initial condition, the time response of $H(t)$ and $h(t)$ increases to be positive before the sudden disturbance occurs. At $t = 100\text{s}$, the huge external disturbance heavily perturbs the system, making both $H(t)$ and $h(t)$ negative, as shown in the zoom-in subfigures of Figure \ref{PAPC_HSd}. Nevertheless, with the proposed PAP control scheme, the system is still able to be guided toward the performance-satisfied region, and finally achieving the desired PAP requirements eventually as $H(t)>0$ and $h(t)>0$ hold.

In this control scenario, a sudden disturbance is exerted on the system. For typical PPC control schemes, if the state trajectory is pulled out of the performance envelope by the perturbations, then it will fail to recover to the performance-satisfied region again. Also, such a process is often accompanied by chattering, or even the divergence of states, due to the well-known singularity problem. Such a problem has been mentioned and discussed in much literature. However, our proposed method is still able to ensure the fulfillment of the PAP requirements even under severe non-ideal conditions.

%\subsection{Comparison with Typical PPC Controllers}

\subsection{Monte Carlo Simulation}
In this subsection, we present a Monte Carlo simulation campaign considering random initial phase of $\boldsymbol{q}_{ev}$. 

In order to provide a smooth and fast convergence behavior, the RPF is specified bys $\rho(0) = \boldsymbol{q}_{ev}(0)$; $T_{sd} = 50\text{s}$, and $\rho_{\infty} = 0$, the width of the tube is the same as in Subsection \ref{normalcase}. 
It is expected that the state trajectory is able to converge to the steady-state region accurately at $t=T_{sd}$, while achieving the desired accuracy requirements.
\begin{figure}[hbt!]
	\centering 
	\includegraphics[width=0.5\textwidth]{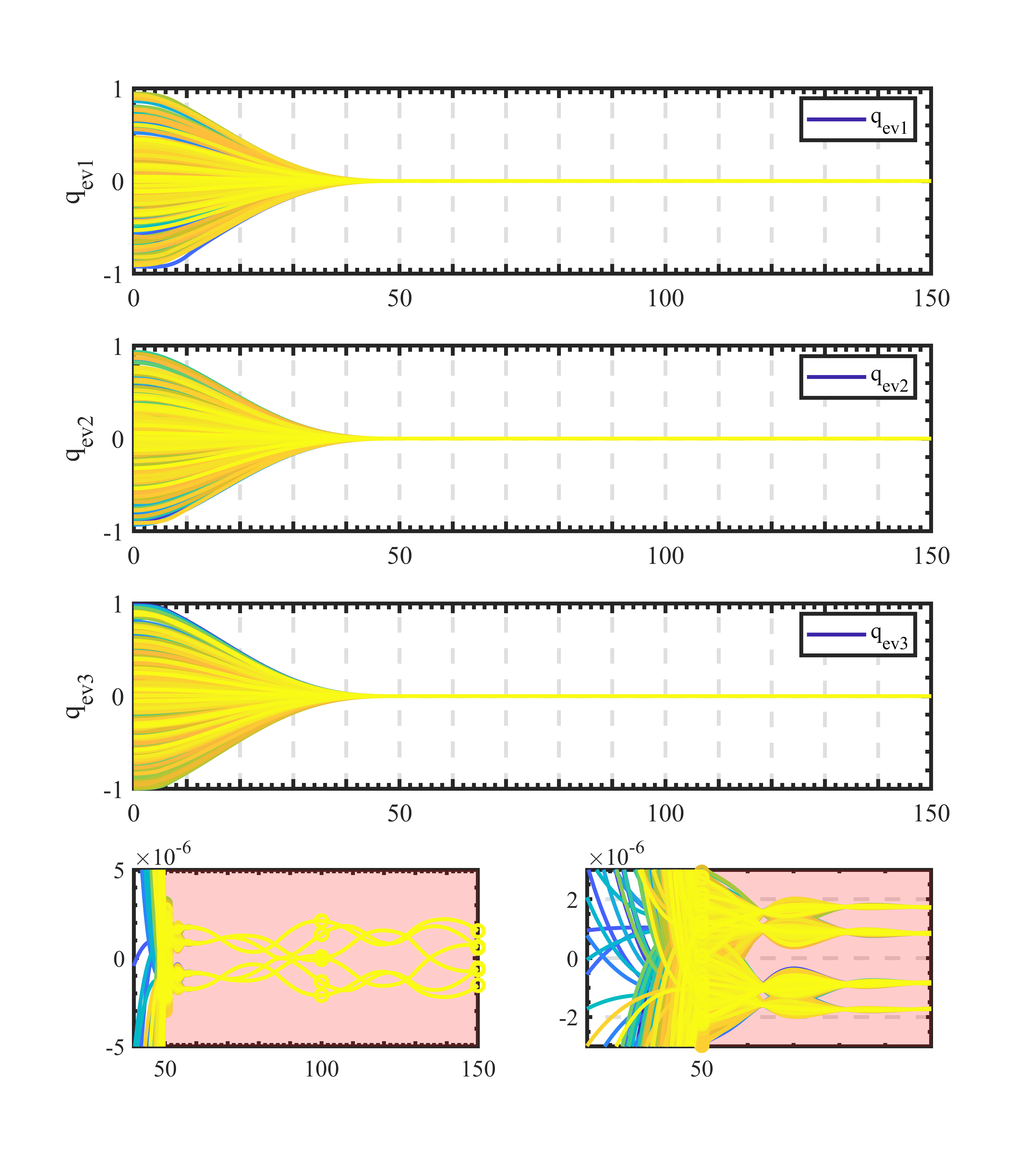}
	\caption{Time evolution of $\boldsymbol{q}_{ev}$ of all cases (Monte Carlo)}    
	\label{QEMC} 
\end{figure}

The simulation results are presented in Figures \ref{QEMC} and \ref{HhMC}. Figure \ref{QEMC} illustrates the time evolution of $\boldsymbol{q}_{ev}$ for all 1500 simulation cases, while the zoom-in figure shows the time evolution near $t=T_{sd}$, and its steady-state control accuracy. Figure \ref{HhMC} illustrates the time evolution of $H(t)$ and $h(t)$ for all cases.

From Figure \ref{QEMC}, it can be observed that all simulation cases are able to converge to satisfy $|q_{evi}(t)| < \Delta_{e} = 1\times 10^{-5}$ just near the assigned time instant $T_{sd} = 50\text{s}$, while its control accuracy is greater than $\max\left(q_{evi}(t)\right)<2\times 10^{-6}$, indicating that the performance criteria are achieved. Also, there is only tiny overshooting occurring during the converging process. From Figure \ref{HhMC}, note that the values of $h(t)$ and $H(t)$ remain to satisfy $h(t)>0$ and $H(t)>0$ for all cases, showing the capability of the proposed control scheme when facing various conditions.

\begin{figure}[hbt!]
		\centering 
	\includegraphics[width=0.5\textwidth]{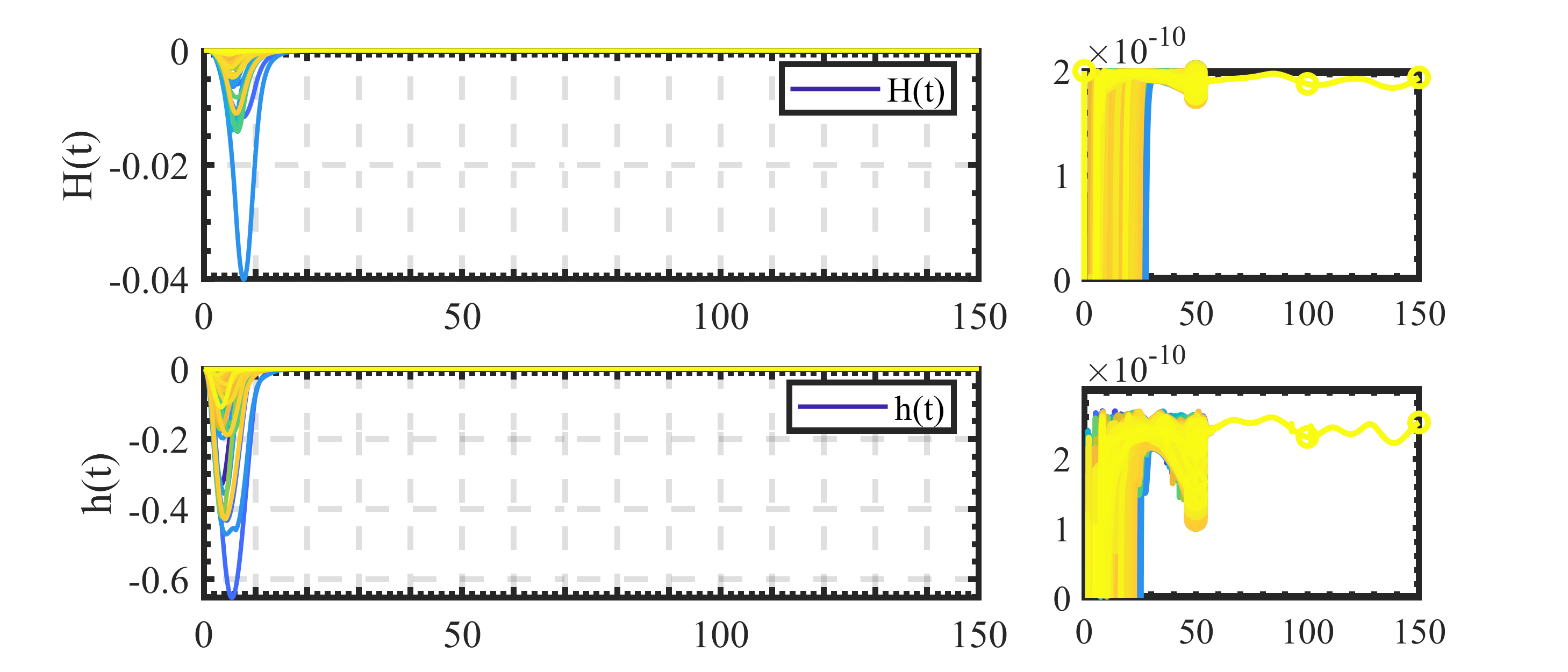}
	\caption{Time evolution of $H(t)$ and $h(t)$ of all cases (Monte Carlo)}    
	\label{HhMC} 
\end{figure}

\section{Conclusions}
This paper investigates the attitude tracking control problem of spacecraft, aiming to achieve precisely-assigned performance criteria, including precise settling time, precisely-assigned steady-state error bound, and overshoot elimination. To address this issue, we proposed a PAP control scheme, the key logic of which is to guide state trajectories into performance-satisfied regions centered around designated reference performance functions. The concept of control barrier function and Sontag's universal formula for stabilization are introduced to lead the system to satisfy the necessary condition for realizing the attraction of error states, providing an easy-to-implement controller structure that can be regarded as a backstepping controller with time-varying gain. Simulation results have validated the effectiveness of the proposed PAP control scheme, which is able to re-stabilize and achieve precisely-assigned performance criteria even under severe perturbations.

\section{Appendix: Convergence Proof of the Observer Design}\label{append}

\begin{proof}
	Considering the observer design outlined in equation (\ref{observerdesign}), we derive its resulting error dynamics as follows:
	\begin{equation}\label{observererror}
		\begin{aligned}
			\dot{\boldsymbol{e}}_{1} &= -C_{1}\beta\boldsymbol{e}_{1} + \boldsymbol{e}_{2},\quad
			\dot{\boldsymbol{e}}_{2} = -C_{2}\beta^{2}\boldsymbol{e}_{1}-\boldsymbol{\xi}
		\end{aligned}
	\end{equation}
	Introducing an auxiliary state $\boldsymbol{\varepsilon} = \left[\beta\boldsymbol{e}^{\text{T}}_{1},\boldsymbol{e}^{\text{T}}_{2}\right]^{\text{T}}\in\mathbb{R}^{6}$, the error system (\ref{observererror}) can be rearranged as $\dot{\boldsymbol{\varepsilon}} = \boldsymbol{P}_{2}\boldsymbol{\varepsilon}+\boldsymbol{G}$, where $\boldsymbol{P}_{2}\in\mathbb{R}^{6\times 6}$ and $\boldsymbol{G}\in\mathbb{R}^{6}$ are defined as:
	\begin{equation}
		\boldsymbol{P}_{2}\triangleq
		\begin{bmatrix}
			-C_{1}\boldsymbol{I}_{3} & \boldsymbol{I}_{3}\\
			-C_{2}\beta\boldsymbol{I}_{3} & \boldsymbol{0}_{3}
		\end{bmatrix},\quad 
		\boldsymbol{G}\triangleq 
		\begin{bmatrix}
			0\\-\boldsymbol{\xi}
		\end{bmatrix}
	\end{equation}
	Since $\boldsymbol{P}_{2}$ is Hurwitz with $C_{1}>0$, $C_{2}>0$, and $\beta>0$, a positive-definite symmetric matrix $\boldsymbol{Q}_{2} = \boldsymbol{Q}^{\text{T}}_{2}>0$ can be found such that $\boldsymbol{P}^{\text{T}}_{2}\boldsymbol{Q}_{2}+\boldsymbol{Q}_{2}\boldsymbol{P}_{2}=-c_{P}\boldsymbol{I}_{6}$ for some $c_{P}>0$.
	
	Considering a Lyapunov candidate $V_{\varepsilon} = \boldsymbol{\varepsilon}^{\text{T}}\boldsymbol{Q}_{2}\boldsymbol{\varepsilon}$, its time-derivative yields:
	\begin{equation}\label{V_eps}
		\begin{aligned}
			\dot{V}_{\varepsilon} 
			&\le -c_{P}\boldsymbol{\varepsilon}^{\text{T}}\boldsymbol{\varepsilon} + 2\lambda_{Q2\max}\|\boldsymbol{\varepsilon}\|\|\boldsymbol{\xi}\|
		\end{aligned}
	\end{equation}
	where $\lambda_{Q2\max}$ denotes the maximum eigenvalue of $\boldsymbol{Q}_{2}$.
	Utilizing the Young's Inequality, we obtain:
	\begin{equation}\label{dotV3}
		\begin{aligned}
			\dot{V}_{\varepsilon} &\le -\frac{c_{P}}{\lambda_{Q2\max}}V_{\varepsilon}+2\lambda_{Q2\max}\|\boldsymbol{\varepsilon}\|\|\boldsymbol{\xi}\|\\
			&\le -\left[\frac{c_{P}}{\lambda_{Q2\max}}-\frac{\lambda^{2}_{Q2\max}\zeta^{2}_{m}}{h\lambda_{Q2\min}}\right]V_{\varepsilon}+h_{m}
		\end{aligned}
	\end{equation}
	where $h_{m}>0$ is a small positive constant, $\zeta_{m}>0$ is a positive constant satisfying $\|\boldsymbol{\xi}\|\le \zeta_{m}$. Regarding Assumption \ref{Assd}, since we have assumed $\|\dot{\boldsymbol{d}}\| \le \xi_{d}$, $\zeta_{m}$ can be given accordingly.
	
	By choosing $\boldsymbol{Q}_{2}$ appropriately, one can ensure $\dot{V}_{\varepsilon} \le -C_{d}V_{\varepsilon}+h_{m}$, where $C_{d}>0$ is a design constant.
	Consequently, $V_{\varepsilon}$ converges exponentially, and eventually converges into a compact set $\Omega_{\varepsilon}$, defined as:
	\begin{equation}
		\Omega_{\varepsilon} \triangleq \left\{\boldsymbol{\varepsilon}\quad|\quad\|\boldsymbol{\varepsilon}\|<\sqrt{\frac{h_{m}}{\chi C_{d}\lambda_{Q2\min}}}\right\}
	\end{equation}
	where $\chi\in\left(0,1\right)$ a small constant.
	From $\boldsymbol{\varepsilon}$, it follows that $\|\boldsymbol{e}_{2}\|\le \|\boldsymbol{\varepsilon}\|$, therefore the estimation error satisfies $\|\tilde{\boldsymbol{d}}\| \le \lambda_{J\max}\|\boldsymbol{e}_{2}\| \le \lambda_{J\max}\|\boldsymbol{\varepsilon}\|$, and we have:
	\begin{equation}\label{Dedef}
		\|\tilde{\boldsymbol{d}}\| \le \lambda_{J\max}\sqrt{\frac{h_{m}}{\chi C_{d}\lambda_{Q2\min}}} \triangleq D_{e}
	\end{equation}
	This indicates the boundedness of the observer estimation error and thereby proves the convergence of the observer.
\end{proof}

\bibliographystyle{IEEEtran}
\bibliography{IEEEabrv,PAPCref.bib}
\end{document}